%% file: wpnw.tex
\documentclass[
a4paper,
UKenglish,
envcountsame,
]{llncs}

\usepackage[utf8]{inputenc}
\usepackage[textsize=scriptsize,obeyFinal]{todonotes}

\usepackage[T1]{fontenc}
\usepackage{color}

\usepackage{amsmath, amsfonts, amssymb, dsfont}
\usepackage{enumerate}

\usepackage{thmtools}
\usepackage{wrapfig}
\include{shortcuts}

\include{macros-wpnw}

\usetikzlibrary{positioning, shapes,patterns,shadows.blur, arrows,decorations,arrows,automata,shadows,patterns,chains,graphs,calc,intersections,matrix,fit,shapes,chains,decorations.pathreplacing}
\usepackage{paralist}

\definecolor{cau}{RGB}{156,10,125} % HEX #9C0A7D

\spnewtheorem*{appendixlemma}{Lemma}{\bfseries}{}

\newif\ifpaper
\paperfalse % SET FLASE FOR ARCHIVE VERSION

\title{Weighted \PN{} Words: Mind the Gap}
\titlerunning{Weighted \PN{} Words: Mind the Gap}

\author{Yannik Eikmeier\inst{1} \and 
Pamela Fleischmann\inst{1}\thanks{Supported by DFG grant 437493335} \and
Mitja Kulczynski\inst{1} \and
Dirk Nowotka\inst{1}}

\institute{Kiel University, Germany, \email{stu204329@mail.uni-kiel.de, 
$\{$fpa,mku,dn$\}$@informatik.uni-kiel.de}}

\authorrunning{Y. Eikmeier, \and P. Fleischmann, \land M. Kulczynski \and D. Nowotka}

\begin{document}

\maketitle

\begin{abstract}
\input{abstract}
\end{abstract}

\section{Introduction}
%%%%%%%%%%%%%%%%%%%%%%%%%%%%%%%%%%%%%%%%%%%%%%%%%%%%%%%%%%%%%%%%%%%%%%
\input{intro}

\section{Preliminaries}
%%%%%%%%%%%%%%%%%%%%%%%%%%%%%%%%%%%%%%%%%%%%%%%%%%%%%%%%%%%%%%%%%%%%%%
\label{sec:prelims}
\input{prelims}

\section{Weighted \PN{} Words and Weighted \PN{} Form}
%%%%%%%%%%%%%%%%%%%%%%%%%%%%%%%%%%%%%%%%%%%%%%%%%%%%%%%%%%%%%%%%%%%%%%
\label{sec:wpn}
\input{section-wpnw}

\section{Gapfree \WM s}
%%%%%%%%%%%%%%%%%%%%%%%%%%%%%%%%%%%%%%%%%%%%%%%%%%%%%%%%%%%%%%%%%%%%%%
\label{sec:gwm}
\input{sec-gwm}

% \section{Injective \WM s}
% %%%%%%%%%%%%%%%%%%%%%%%%%%%%%%%%%%%%%%%%%%%%%%%%%%%%%%%%%%%%%%%%%%%%%%
% \label{sec:iwm}
% \input{sec-iwm}
\section{Conclusions}
%%%%%%%%%%%%%%%%%%%%%%%%%%%%%%%%%%%%%%%%%%%%%%%%%%%%%%%%%%%%%%%%%%%%%%
\input{conclusion}

\bibliography{refs}

 \newpage

\appendix
\section{Further Insights}
%%%%%%%%%%%%%%%%%%%%%%%%%%%%%%%%%%%%%%%%%%%%%%%%%%%%%%%%%%%%%%%%%%%%%%
\label{sec:further}
\input{furtherinsights}

\clearpage
\ifpaper % ONLY PAPER
	\section{Proofs}
	\label{sec:proof}
    \input{appendix-proofs}

 \else    % ONLY ARCHIVE
 \fi
 % IN BOTH

\end{document}

%% file: shortcuts.tex
\def\ta{\mathtt{a}}
\def\tb{\mathtt{b}}
\def\tc{\mathtt{c}}

\def\tx{\mathtt{x}}
\def\ty{\mathtt{y}}

\def\tn{\mathtt{n}}

\def\tzero{\mathtt{0}}
\def\tone{\mathtt{1}}

\def\etal{\emph{et al. }}

\def\nth#1{#1$^{\text{th}}$}

\def\abs#1{|#1|}

\renewcommand{\epsilon}{\varepsilon}
\renewcommand{\phi}{\varphi}

\def\N{\mathbb{N}} %natural numbers
\def\Z{\mathbb{Z}} %natural numbers
\def\P{\mathbb{P}} %prime numbers

\DeclareMathOperator{\Fact}{Fact}
\DeclareMathOperator{\Pref}{Pref}

\newcommand*\colvec[1]{\bigl(\begin{smallmatrix}#1\end{smallmatrix}\bigr)}

\usepackage{bbm}
\def\neutr{\mathbbm{1}} %neutral element

\DeclareMathOperator{\maxpos}{maxpos}
\DeclareMathOperator{\minpos}{minpos}
\DeclareMathOperator{\pos}{pos}

\DeclareMathOperator{\dom}{dom}

%% file: macros-wpnw.tex
\def\pn{prefix normal}
\def\Pn{Prefix normal}
\def\PN{Prefix Normal}

\def\pnity{\pn ity}
\def\Pnity{\Pn ity}

\def\factXf{maximum-$X$-factor function}

\def\prefXf{prefix-$X$-function}

\def\maxposWf{max-position function}

\def\minposWf{min-position function}

\def\factWf{factor-weight function}

\def\prefWf{prefix-weight function}

\def\factorWEqnce{factor-weight equivalence}
\def\factorWEqnt{factor-weight equivalent}

\def\wm{weight measure}
\def\WM{Weight Measure}

\def\sumwm{sum \wm}
\def\prodwm{product \wm}
\def\pwm{prime \wm}

\def\ficiliptak{Fici and Lipt\'{a}k}

%% file: abstract.tex
A \pn{} word is a binary word whose prefixes contain at
least as many 1s as any of its factors of the same length.
Introduced by Fici and Lipták in 2011 the notion of \pnity{}
is so far only defined for words over the binary alphabet.
In this work we investigate a generalisation for 
finite words over arbitrary finite alphabets,
namely weighted \pnity.
We prove that weighted \pnity{} is more expressive 
than binary \pnity{}. Furthermore, we investigate
the existence of a weighted \pn{} form since 
weighted \pnity{} comes with several new peculiarities 
that did not already occur in the binary case.
We characterise these issues and finally present a standard 
technique to obtain a generalised \pn{} form for all words over
arbitrary, finite alphabets.

%% file: intro.tex
Complexity measures of words
are a central topic of investigation when 
dealing with properties of sequences, e.g. 
factor complexity 
\cite{BernatMP07,BalaziMP07,BUCCI200960,SHALLIT201996},
binomial complexity \cite{RigoS15,FreydenbergerGK15,LeroyRS17a,Rigo19}, 
cyclic complexity \cite{CassaigneFSZ14}. 
Characterising the maximum density of a 
particular letter in the set of factors of a given length,
hence considering an abelian setting,
falls into that category (see for instance 
\cite{CassaigneKS17,RichommeSZ11,Blanchet-SadriS16}
 and the references therein).
Such characterisations inevitably prompt
the search for and investigation of normal
forms representing words with equivalent measures.
\Pnity{} is the concept considered in this 
paper and was first introduced by Fici and 
Lipták in 2011 \cite{Fici2011} as a property 
describing the distribution of a designated letter 
within a binary word. 
A word over the binary alphabet $\{\tzero,\tone\}$ is 
\emph{\pn} if its prefixes contain 
at least as many $\tone$s as any of its factors of
the same length. For example, the word $\mathtt{1101001}$ 
is \pn{}. 
Thus, prefixes of \pn{} words give an upper bound for the amount of $\tone$s
any other factor of the word may contain.
For a given binary word $w$ the \emph{maximum-$\tone$s function} 
maps $n$ to the maximum amount of $\tone$s, 
a length-$n$ factor of $w$  can have.
In~\cite{PNF} Burcsi~\etal show that there exists 
exactly one \pn{} word (the \pn{} form) 
in the set of all binary words that have 
an identical maximum-$\tone$s function, e.g. the \pn{} form of
$\mathtt{1001101}$ is $\mathtt{1101001}$.

From an application point of view this
complexity measure is directly connected to the 
{\em Binary Jumbled Pattern Matching Problem (BJPM)}
(see e.g.~\cite{IJPM,AJPM,ABBY} and for the general JPM, see e.g. 
\cite{DBLP:conf/esa/KociumakaRR13}).
The BJPM problem is to determine whether a given finite binary 
word has factors containing given amounts of $\tone$s and $\tzero$s.
In~\cite{Fici2011} \pn{} forms were used to construct 
an index for the BJPM problem in $O(n)$ time 
where $n$ is the given word's length. 
The fastest known algorithm for this problem
has a runtime of $O(n^{1.864})$ (see~\cite{CLUSTER}).
In \cite{balister2019asymptotic} Balister and Gerke showed that 
the number of length-$n$ \pn{} words is $2^{n-\Theta(\log^2(n))}$,
and the class of a given \pn{} word contains at most
$2^{n-O(\sqrt{n\log(n)})}$ elements. 
In more theoretical settings, 
the language of binary \pn{} words has 
also been extended to infinite binary words~\cite{IPNW}.
\Pnity{} has been shown to be connected to other 
fields of research within combinatorics on words,
e.g. Lyndon words~\cite{Fici2011} and
bubble languages~\cite{GeneratingPNW}.
Furthermore, efforts have been made
to recursively construct \pn{} words,
via the notions of extension 
critical words (collapsing words) and \pn{} palindromes
~\cite{CollPNW,GeneratingPNW,DBLP:journals/tcs/CicaleseLR18}. 
The goal therein was
to learn more about the 
number of words with the 
same \pn{} form and the number 
of \pn{} palindromes.
Very recently in \cite{DBLP:journals/corr/abs-2003-03222}
a Gray code for \pn{} words in amortized polylogarithmic
time per word was generated.
Four sequences related to \pn{} words can be found in the
OEIS \cite{OEIS}:
A194850 (number of \pn{} words of length $n$),
A238109 (list of \pn{} words over the binary alphabet),
A238110 (maximum number of binary words of 
length $n$ having the same \pn{} form), and 
A308465 (number of \pn{} palindromes of length $n$).

\textbf{Our contribution.}
In this work, we investigate a generalisation of \pnity{}
for finite words over arbitrary finite alphabets.
We define a \emph{\wm}, which is a morphic function assigning a weight 
(an element from an arbitrary but a priori chosen monoid)
to every letter of an arbitrary finite alphabet.
Based on those weights we can again compare
factors and prefixes of words over this alphabet w.r.t. their weight.
A word is {\em \pn{} w.r.t. a \wm{}} if no factor has a higher weight than that 
of the prefix of the same length.
Note here, for some \wm s not every word has a unique \pn{} form. 
We prove basic properties of \wm s and weighted \pnity{}
and give a characterisation of \wm s for which every word has a \pn{} form.
Finally, we define a standard \wm{} which only depends on the alphabetic
order of the letters and a unique \emph{weighted \pn{} form} that is not 
depending on the choice of a \wm.

%Finally we briefly discuss a na\"ive approach to generalise
%binary \pnity{}: subset prefix normality.
%We prove that this approach can already be 
%obtained by the use of a \wm{} with binary weights.

\textbf{Structure of the paper.}
In Section~\ref{sec:prelims}, we define the basic terminology.
Following that, in Section~\ref{sec:wpn}, we prove that weighted \pnity{}
is a proper generalisation of the binary case and present our results on 
the existence of a weighted \pn{} form.
Finally, in Section~\ref{sec:gwm}, we present our main theorem on the standard 
\wm{} as well as the weighted \pn{} form.

\ifpaper % ONLY PAPER
Due to space restrictions, the proofs can be found in the 
Appendix~\ref{sec:proof}.
Moreover, some further insights about weighted \pnity{}, e.g. a second but less 
powerful approach and some basic properties related to binary \pnity{}, are 
also given in Appendix~\ref{sec:further}.
A full version of this work can be found at \cite{wpnwarxiv}.
\else    % ONLY ARCHIVE
Some further insights about weighted \pnity{}, e.g. a second but less 
powerful approach and some basic properties related to binary \pnity{}, 
are also given in the Appendix~\ref{sec:further}.
\fi

%% file: prelims.tex
%!TEX root = wpnw.tex 
Let $\N$ denote the positive natural numbers $\{1, 2, 3, \dots\}$, 
$\Z$ the integer numbers, and $\P\subset\N$ the set of prime numbers.
Set $\N_0 := \N \cup \{0\}$. For $i,j\in\N$, we define the interval
$[i,j] := \{n\in\N\mid i\leq 
n\leq j \}$ and
for $n\in\N$, we define $[n] := [1,n]$ and $[n]_0 := [0,n]$.
For two monoids $A$ and $B$ with operations $\ast$ and $\circ$ respectively,
a function $\mu : A \to B$ is a \emph{morphism} if 
$\mu(x*y) = \mu(x) \circ \mu(y)$ holds for all $x,y\in A$.
Notice, if the domain $A$ is a free monoid
over some set $S$, a morphism from $A\to B$ 
is sufficiently defined by giving a mapping from $S$ to $B$.

An \emph{alphabet} $\Sigma$ is a finite set of letters. 
A \emph{word} is a finite sequence of letters from a given alphabet.
Let $\Sigma^*$ denote the set of all finite words over $\Sigma$, 
i.e. the free monoid over $\Sigma$.
Let $\epsilon$ denote the \emph{empty word} 
and set $\Sigma^+ := \Sigma^* \backslash \{\epsilon\}$ 
as the free semigroup over $\Sigma$.
We denote the length of a word $w\in\Sigma^*$ by $|w|$, 
i.e. the number of letters in $w$.
Thus $|\epsilon| = 0$ holds.
Let $w$ be a word of length $n\in\N$.
Let $w[i]$ denote the \nth{$i$} letter of 
$w$ for $i\in[|w|]$, 
and set $w[i\dots j] = w[i]\cdots w[j]$ for $i,j\in[|w|]$ and $i\leq j$.
Let $w[i\dots j]=\epsilon$ if $i>j$.
The number of occurrences of a letter $\ta \in \Sigma$ 
in $w$ is denoted by 
$|w|_\ta = \abs{\{i\in[|w|] \mid w[i]=\ta\}}$.
We say $x \in \Sigma^*$ 
is a \emph{factor} of $w$ 
if there exist $u,v\in\Sigma^*$ with $w = uxv$.
In this case $u$ is called a \emph{prefix} of $w$.
We denote the set of $w$'s factors (resp. prefixes) by 
$\Fact(w)$ (resp. $\Pref(w)$) and $\Fact_i(w)$ ($\Pref_i(w)$ resp.)
denotes the set of factors (prefixes) of length $i\in[|w|]$.
Given a total order $<$ over $\Sigma$ let $<_{lex}$ denote 
the extension of $<$ to a lexicographic order over $\Sigma^*$.
Fixing a strictly totally ordered alphabet $\Sigma = \{\ta_1,\ta_2,\dots,
\ta_n\}$ with $\ta_i<\ta_j$ for $1 \leq i<j \leq n$,
the \emph{Parikh vector} of a word is defined by
$p: \Sigma^* \to \N^n : w \mapsto 
(|w|_{\ta_1},|w|_{\ta_2},\dots,|w|_{\ta_n})$.
For a function $f$ set $f(A)=\{f(a) \mid a \in A\}$ for $A\subseteq\dom(f)$.

Before we define the \wm s and weighted \pnity{}
we recall the definition for binary \pnity{}
as introduced by \ficiliptak{} in~\cite{Fici2011}.

\begin{definition}(\cite{Fici2011})
    \label{def:bpn}
    Given $w\in\{\tzero, \tone\}^*$ the \emph{maximum-ones function} $f_w$ and 
    the \emph{prefix-ones function} $p_w$ are respectively defined by
    $f_w :\;[|w|]_0\to\N_0,\;i\mapsto\max(|\Fact_i(w)|_\tone)$
     and
    $p_w :\;[|w|]_0 \to \N_0,\;i \mapsto |\Pref_i(w)|_\tone$.
    The word $w$ is called \emph{\pn{}}
    if $f_w=p_w$ holds.
\end{definition}

Our generalisation of binary \pnity{} is based on so called {\em weight 
measures}, i.e. we apply weights represented by elements from a strictly 
totally ordered monoid $A$ to every letter of the alphabet. In the 
following we denote the neutral element of an arbitrary monoid $A$ by 
$\neutr_A$, its operation by $\circ_A$, and its total order by $<_A$
(in the case of existence).

\begin{definition}
    \label{def:wm}
    Let $A$ be a totally ordered monoid.
    A morphism $\mu : \Sigma^* \to A$ is a 
    \emph{\wm} over the alphabet $\Sigma$ 
    w.r.t. $A$ if $\mu(vw) = \mu(wv)$ 
    and $\mu(w) <_A \mu(wv)$ hold for all words
    $w\in\Sigma^*$ and $v\in\Sigma^+$.
    We refer to the second property as the {\em increasing property}.
    We say the weights of the letters of $\Sigma$ 
    are the \emph{base weights} of $\mu$, so $\mu(\Sigma)$ is the
    set of all \emph{base weights}.
\end{definition}

\begin{remark}
    Notice that if there exists a \wm{} $\mu$ 
    w.r.t. the monoid $A$ then
    $|A|$ is infinite, $\circ_A$ is commutative on $\mu(\Sigma^*)$,
    and $\mu(\epsilon) = \neutr_A$ holds.
    Moreover, the increasing property of \wm s ensures
    that only the neutral element $\varepsilon$ of $\Sigma^{\ast}$
    is mapped to the neutral element $\neutr_A$.
    Hence, we will see that our factor- and \prefWf s
    are strictly monotonically increasing
    in contrast to the functions defined in \cite{Fici2011}.
    However, if we allow letters from $\Sigma$ to  be also assigned
    the neutral weight $\neutr_A$,
    we get the known results for binary alphabets.
\end{remark}

\begin{remark}  
    Notice, that a \wm{} $\mu$ can be defined for any alphabet $\Sigma$
    in two steps: choose some infinite commutative monoid 
    with a total and strict order and  assign a base weight that is greater 
    than the neutral element to 
    each letter in $\Sigma$. Since $\mu$ is a morphism, the weight of a word 
    $w\in\Sigma^{\ast}$ is well defined.
\end{remark}

In the following definition we introduce seven (for us the most intuitive) special types of \wm s.

\begin{definition}
    A weight measure $\mu$ over the alphabet $\Sigma$ w.r.t. the monoid $A$ is
    \begin{description}
            \setlength{\itemsep}{-6pt}%
            \setlength{\parskip}{-6pt}%
        \item{$\diamond$ \emph{injective}} if $\mu$ is injective on $\Sigma$,\\
        \item{$\diamond$ \emph{alphabetically ordered}} if $\mu(\ta) \leq_A \mu(\tb)$ holds for all 
    $\ta,\tb\in\Sigma$ 
with $\ta \leq \tb$ for a total order $\leq$ on $\Sigma$,\\
    \item{$\diamond$  \emph{binary}}
        if $|\mu(\Sigma)|=2$ holds, 
        and \emph{non-binary} if $|\mu(\Sigma)|>2$,\\
    \item{$\diamond$ \emph{natural}}
        if $A$ is $\N_0$ or $\N$ with $<_A$ 
        being the usual order $<$ on integers,\\
    \item{$\diamond$ a \emph{\sumwm{}}}
        if it is natural and the operation on $A$ is $+$,\\
    \item{$\diamond$ a \emph{\prodwm{}}} if it is natural and the operation on $A$ is $*$,\\
    \item{$\diamond$ \emph{prime}} if it is a \prodwm{}
        and $\mu(\Sigma) \subseteq \P$ holds.
    \end{description}

\end{definition}

Consider, for instance, the alphabet $\Sigma=\{\ta,\tb,\tc\}$.
The \wm{} $\mu$ over $\Sigma$ 
with $\mu(\ta) = 1$,
$\mu(\tb) = 2$, and $\mu(\tc) = 3$
is \emph{non-binary}, \emph{natural}, and 
with the monoid $(\N_0,+)$ it is a \emph{\sumwm{}}. 
It cannot be a \emph{\prodwm{}}
with $(\N,*)$ since then $\mu(\ta)=1$ would 
violate the increasing property. However, the \wm{} $\nu$ over $\Sigma$ w.r.t. $(\N,*)$ 
with $\nu(\ta) = 2$,
$\nu(\tb) = 3$, and $\nu(\tc) = 5$
is not only a \emph{\prodwm{}},
but also a \emph{\pwm{}} (for further insights regarding \pwm s see 
Subsection~\ref{subsec:pwms}).

\begin{remark}    
    For the binary alphabet $\Sigma=\{\tzero,\tone\}$
    a \sumwm{} $\mu$ with $\mu(w)=|w|_\tone$ 
    for all $w\in\Sigma^*$ cannot exist since we would have 
    $\mu(\tzero) = 0 = \mu(\epsilon)$
    which is a contradiction to the increasing property. 
    Later on we are going to circumvent this problem
    by setting $\mu(w)=|w|_1+|w|$
    for all $w\in\Sigma^*$ when 
    implementing binary \pnity{}
    via the usage of weight measures.
    Alternatively, we may relax the 
    increasing property and allow $\mu(0)=0$;
    this results exactly in the same properties
    as discussed in \cite{Fici2011}.
\end{remark}

We now define the analogues to the maximum-ones and prefix-ones functions.
    
%Def f and p
\begin{definition}
    \label{def:Wfs}
    Let $w\in\Sigma^*$ and $\mu$ be a \wm{} 
    over the alphabet $\Sigma$ w.r.t. the monoid $A$.
    Define the \emph{\factWf} $f_{w,\mu}$
    and \emph{\prefWf} $p_{w,\mu}$ respectively by
    $
    f_{w,\mu} : \; [|w|]_0 \to A,\;
    i\mapsto \max(\mu(\Fact_i(w)))
    $ and 
    $
    p_{w,\mu} : \; [|w|]_0 \to A,\;
    i\mapsto \mu(\Pref_i(w))
    $.
\end{definition}

For instance, let $\mu$ be a \sumwm{} with the base weights
$\mu(\ta) = 1$, $\mu(\tn) = 2$, $\mu(\tb) = 3$
for the alphabet $\Sigma=\{\ta,\tn,\tb\}$.

Now consider the words $\mathtt{banana}$ and $\mathtt{nanaba}$.
Table~\ref{tab:banana} shows the mappings of

\begin{wraptable}[8]{L}{.45\linewidth}
    \vspace*{-0.5cm}
    \begin{tabular}{c|cccccc}
        $i$ &1&2&3&4&5&6 \\\hline            
        $p_{\mathtt{nanaba},\mu}(i)$&2&3&5&6&9&10 \\
        $f_{\mathtt{nanaba},\mu}(i)$&3&4&6&7&9&10 \\\hline
        $p_{\mathtt{banana},\mu}(i)$,
        $f_{\mathtt{banana},\mu}(i)$&3&4&6&7&9&10
    \end{tabular}

    \caption{Comparing $\mathtt{banana}$'s and $\mathtt{nanaba}$'s
        prefix- and factor-weights.%\factWf s.
        \label{tab:banana}}
\end{wraptable}

\noindent their
prefix- and \factWf s.
The \factWf{} of $\mathtt{nanaba}$ is realised by the factors
$\mathtt{b}$, $\mathtt{ab}$, $\mathtt{nab}$, $\mathtt{anab}$,
$\mathtt{nanab}$, $\mathtt{nanaba}$.

Finally, we  define a generalised approach for \pnity{},
namely the {\em weighted \pnity{}} for a given weight measure $\mu$.
As in the binary case, for a \pn{} word the \factWf{} 
and the \prefWf{} have to be identical.

%Def muPN
\begin{definition}
\label{def:gpn}
    Let $w\in\Sigma^*$ and let $\mu$ be a \wm{} over $\Sigma$.
    We say $w$ is \emph{$\mu$-\pn{}} (or weighted \pn{} w.r.t. $\mu$)
    if $p_{w,\mu} = f_{w,\mu}$ holds.
\end{definition}

In the example above we see $p_{\mathtt{banana},\mu}=f_{\mathtt{banana},\mu}$ 
holds and hence $\mathtt{banana}$ is prefix normal w.r.t. $\mu$. On the other 
hand we have $p_{\mathtt{nanaba},\mu}(1)=2<3=f_{\mathtt{nanaba},\mu}(1)$ and 
therefore $\mathtt{nanaba}$ is not prefix normal w.r.t. $\mu$.

%% file: section-wpnw.tex
%!TEX root = wpnw.tex 
In this section we show that \emph{weighted \pnity} is a proper generalisation 
of binary \pnity{} and further investigate the weighted \pn{} form. 
By examining special properties of \wm s, we intend to guide the reader from 
the general approach to a characterisation of special \wm s for which 
every word has a weighted \pn{} equivalent, namely \emph{injective and gapfree 
\wm s}.
(Some useful basic properties  which are direct generalisations of the 
binary case can be found in Subsection~\ref{subsec:wposf}.)
Before we define the analogue to the prefix-equivalence for factor weights,
we show that weighted \pnity{} is more general and more expressive than binary 
\pnity{}, i.e. every statement on binary \pnity{} can be expressed by
weighted \pnity{} but not vice versa.

% \ifpaper
% \else
% \input{lembasics}
% \fi
% \ifpaper % ONLY PAPER
% \else    % ONLY ARCHIVE 
%     \input{proof_lemBasics}
% \fi
% 
% \ifpaper
% \else
% \input{lemf1}
% 
% \fi
% \ifpaper % ONLY PAPER
% \else    % ONLY ARCHIVE 
%     \input{proof_lemF1}
% \fi

% \ifpaper
% \else
% \input{propequivpn}
% \fi
% 
% \ifpaper % ONLY PAPER
% \else    % ONLY ARCHIVE 
% \input{proof_propEqivPN}
% \fi

\begin{proposition}
    \label{the:weight2bin}
    %Weighted \pnity{} is a generalisation of binary \pnity{}.
    Binary \pnity{} is expressible by
%    the concept of
    weighted \pnity{}, i.e. there exists a \wm{} $\mu$ such that 
    $\mu$-\pnity{} is equivalent to binary \pnity{}.
\end{proposition}
\ifpaper % ONLY PAPER
\else    % ONLY ARCHIVE 
\input{proof_propBinByWeight}
\fi

With the binary \sumwm{} $\mu$ over $\Sigma=\{\tzero,\tone\}$
where $\mu(\tone)=2$ and $\mu(\tzero)=1$, we can transform any statement on binary 
\pnity{} into an analogue in the weighted setting.
For example, for $w=\mathtt{11001101}$ we have $f_w(4)=3$ and $p_w(4)=2$
(so $w$ is not \pn{}) and in the weighted setting  we have 
$f_{w,\mu}(4)=7=f_w(4)+4$ and $p_{w,\mu}(4)=6=p_w(4)+4$; or in general
$f_w(i)=f_{w,\mu}(i)-i$ holds for all $w\in\Sigma^*$ and $i\in[|w|]$.
Therefore, $w$ is $\mu$-\pn{} if and only if it is \pn{}.

\begin{definition}          
    \label{def:factorWEq}
    Let $\mu$ be a \wm{} over $\Sigma$.
    Two words $w, w' \in \Sigma^*$ 
    are \emph{\factorWEqnt{}} w.r.t. $\mu$ (denoted by $w \sim_\mu w'$)  if
    $f_{w,\mu} = f_{w',\mu}$ holds. We denote the equivalence classes by 
    $[w]_{\sim_\mu} := \{w'\in\Sigma^* \mid w \sim_\mu w'\}$.
\end{definition}

In the following we highlight three peculiarities about the \factorWEqnce{}
that do not occur in the binary case:
the existence of \factorWEqnt{} words with \emph{different Parikh vectors},
the existence of \emph{multiple} words that are weighted \pn{} and \factorWEqnt{},
and the \emph{absence} of a \factorWEqnt{} word that is weighted \pn{}.
The words $\mathtt{banana}$ and $\mathtt{nanaba}$
over $\Sigma=\{\ta,\tn,\tb\}$ with the weight measure $\mu(\ta)=1$, $\mu(\tn)=2$, and $\mu(\tb)=3$
are factor-weight equivalent.
The complete equivalence class is given by
$[\mathtt{banana}]_{\sim_\mu} =
\{ \mathtt{ananab},\mathtt{anaban},\mathtt{abanan}, 
\mathtt{nanaba},\mathtt{nabana},\mathtt{banana}\}$.
Notice that all words in the class have 
the same Parikh vector
but only $\mathtt{banana}$
is $\mu$-\pn{}.
If we were to add $\tc$ to $\Sigma$
and expand $\mu$ by $\mu(\tc)=\mu(\tn)=2$ then
$[\mathtt{banana}]_{\sim_\mu}$ contains all
words previously in it but also those where
some $\tn$s are substituted by $\tc$.
So $[\mathtt{banana}]_{\sim_\mu}$ contains four
$\mu$-\pn{} words, namely $\mathtt{banana}$,
$\mathtt{bacana}$, $\mathtt{banaca}$, and $\mathtt{bacaca}$.
Lastly, consider the \sumwm{} $\nu$ over the alphabet
$\Sigma=\{\ta,\tn,\tx\}$ with the base weights
$\nu(\ta) = 1$, $\nu(\tn) = 2$, $\nu(\tx) = 4$.
Now $[\mathtt{xaxn}]_{\sim_\nu}$ only contains $\mathtt{xaxn}$
and its reverse $\mathtt{nxax}$.
Interestingly none of the two words are $\nu$-\pn{},
witnessed by
$f_{\mathtt{xaxn},\nu}=f_{\mathtt{nxax},\nu}=(4,6,9,11)$, $p_{\mathtt{xaxn},\nu}=(4,5,9,11)$, and $p_{\mathtt{nxax},\nu}=(2,6,7,11)$
(the functions are written as sequences for brevity).
In order for a weighted \pn{} word to exist in the class,
a letter with weight $f_{\mathtt{xaxn},\nu}(3)-
f_{\mathtt{xaxn},\nu}(2)=9-6=3$ is missing.
For example with such a letter $\tb$ in $\Sigma$ with
$\nu(\tb) = 3$ the word $\mathtt{xnbn}$ is 
$\nu$-\pn{} and in $[\mathtt{xaxn}]_{\sim_\nu}$.
These examples show that \factorWEqnce{} classes
can contain words with different Parikh vectors,
multiple \pn{} words, and even no \pn{} words at all.
We now investigate the question which \wm s cause such
peculiar equivalence classes and characterise the equivalence classes 
that contain a single weighted \pn{} word, a \emph{normal form},
as it always exists for the binary case (see \cite{Fici2011}). 

%We now answer the question which equivalence classes contain
%a single \pn{} word that can 
%be used to represent that class, i.e. a \emph{normal form},
%as it always exists for the binary case ~\cite{Fici2011}. 

\begin{definition}
    \label{def:Pwmu}
    For $w \in \Sigma^*$ and a \wm{} $\mu$ over $\Sigma$ 
    we define the $\mu$-\pn{} subset 
    of the \factorWEqnce{} class of $w$ by
    $ \mathcal{P}_\mu(w) := 
    \{v\in [w]_{\sim_\mu} \mid p_{v,\mu} = f_{v,\mu}\}$.
\end{definition}
    
In the example above, multiple \pn{} words 
in a single class are a direct result of 
ambiguous base weights, i.e.
non-injective \wm: all letters with the same 
weights are interchangeable in any word 
with no effect on the weight of that word;
thus there exist multiple \pn{} words for such a word.
By choosing an injective \wm{} we can avoid this behaviour.
However, the problematic case where some equivalence 
classes contain no \pn{} words at all, still remains.
We give a characterisation of special, so called 
\emph{gapfree}, \wm s and show that they 
guarantee the existence of a \pn{} word in every 
equivalence class of the \factorWEqnce{}.
Before we prove the just stated claims,
we formally define the previous observations of \emph{gaps}.

\begin{definition}
    A \wm{} $\mu$ 
    over the alphabet $\Sigma$ w.r.t. the monoid $A$
    is \emph{gapfree},
    if for all words $w\in\Sigma^*$ and all $i\in[|w|]$ 
    there exists an $\ta\in\Sigma$ such that
    $f_{w,\mu}(i) = f_{w,\mu}(i-1) \circ_A \mu(\ta)$ holds.
    Otherwise, if for any word $w\in\Sigma^*$ and an $i\in [|w|]$ 
    there exists no $\ta\in\Sigma$
    such that $f_{w,\mu}(i)=f_{w,\mu}(i-1) \circ_A \mu(\ta)$ holds
    we say $\mu$ is \emph{gapful} and has a \emph{gap} over 
    the word $w$ at the index $i$.
    \if false
    $$
    \mu \text{ gapfree} \quad :\Leftrightarrow \quad
    \forall w\in\Sigma^*\,\forall i\in[|w|]\,\exists \ta\in\Sigma:\, 
    f_{w,\mu}(i)=f_{w,\mu}(i-1) \circ \mu(\ta)
    $$
    \fi
\end{definition}

Consider for example the \sumwm{}
over $\Sigma =\{\ta,\tb,\tc\}$ with
$\mu(\ta) = 2$, 
$\mu(\tb) = 4$, and
$\mu(\tc) = 6$.
We show that $\mu$ is gapfree by proving
the existence of letters in $\Sigma$ with  
weight $x_i = f_{w,\mu}(i) - f_{w,\mu}(i-1) \in \N$ for all
$w\in\Sigma^*$ and $i\in[|w|]$. Since the \factWf{} is defined as a
maximum, we get $x_i\leq\mu(\tc)=6$.
On the other hand $x_i\geq \mu(\ta)=2$
because the \factWf{} is strictly increasing.
Since all the base weights $\mu(\Sigma)=\{2,4,6\}$ are even,
the same is true for $f_{w,\mu}(i)$ and $f_{w,\mu}(i-1)$.
Thus, $x_i$ has to be even as well.
This implies $x_i\in\{2,4,6\}=\mu(\Sigma)$. Hence, 
there exist letters in $\Sigma$ with the appropriate weight to
fill every possible gap, i.e. $\mu$ is gapfree.
As a counter example, the \sumwm{} $\nu$ over $\Sigma$ with 
$\nu(\ta) = 1$, $\nu(\tb) = 3$, and $\nu(\tc) = 4$
is gapful.
Consider the word $w=\mathtt{bcac}$
then $\nu$ has a gap over $w$ at the index $3$
since $f_{w,\nu}(3)=9$ (witnessed by the factor $\tc\ta\tc$)
and $f_{w,\nu}(2)=7$ (witnessed by the factor $\tb\tc$) and there is no letter 
with weight $2$.

Coming back to the original question of multiple prefix normal words, the 
following theorem characterises exactly when an equivalence class contains no, 
exactly one, or more than one weighted \pn{} word.

\begin{theorem}
    \label{the:P}
    Let $\mu$ be a \wm{} over $\Sigma$. Then \\
    - there exists $w\in\Sigma^{\ast}$ such that 
        $|\mathcal{P}_\mu(w)| = 0$ \mbox{ iff }
        $\mu$ is gapful, \\
    - there exists $w\in\Sigma^{\ast}$ such that 
    $|\mathcal{P}_\mu(w)| > 1$  iff $\mu$ is not injective, and\\
    - for all $w\in\Sigma^{\ast}$ we have $|\mathcal{P}_\mu(w)| = 1$ iff
    $\mu \text{ gapfree and injective}$.
\end{theorem}
\ifpaper % ONLY PAPER
\else    % ONLY ARCHIVE 
\input{proof_theWPNF}
\fi

\begin{definition}
    Let $\mu$ be a gapfree and injective \wm{} over $\Sigma$ 
    and let $w\in\Sigma^*$. Then $|\mathcal{P}_\mu(w)| = 1$  and 
    its element  is the $\mu$-\pn{} form of $w$. 
\end{definition}

Again with the alphabet $\Sigma = \{ \ta, \tn, \tb, \tx\}$ 
and the \sumwm{} $\mu$ over $\Sigma$ with base weights
$\mu(\ta) = 1$, $\mu(\tn) = 2$, $\mu(\tb) = 3$, and $\mu(\tx)=4$
we have
$\mathcal{P}_{\mu}(\mathtt{nanaba}) = \{\mathtt{banana}\}$
and $\mathcal{P}_{\mu}(\mathtt{xaxn}) = \{\mathtt{xnbn}\}$.
So $\mathtt{banana}$ is the $\mu$-\pn{} form of $\mathtt{nanaba}$
and $\mathtt{xnbn}$ is the $\mu$-\pn{} form of $\mathtt{xaxn}$.
%Also if $\tb$ were not in $\Sigma$ we see 
%$\mu$ would have a gap over $\mathtt{xaxn}$
%and $\mathcal{P}_{\mu}(\mathtt{xaxn})$ would be empty.    
Additionally, notice $\mathtt{xaxn}$ is an example 
of a word such that its Parikh vector is different from that of its prefix normal form.

\begin{remark}
    \label{rem:pnfc}
    Let $\mu$ be a gapfree and injective \wm{}
    over the alphabet $\Sigma$ w.r.t. the monoid $A$ 
    and $w\in\Sigma^*$.
    Then the $\mu$-\pn{} form $w'$ of $w$ can be 
    constructed inductively: $w'[1] = \ta$ if $f_{w,\mu}(1) = \mu(\ta)$ and
    for all $i\in[|w|]$, $i>1$ set $w'[i] = \ta \in \Sigma$ if 
    $f_{w,\mu}(i) = f_{w,\mu}(i-1)\circ_A \mu(\ta)$.
    In contrast, for a \wm{} that is gapfree but \emph{not} injective this
    inductive construction can be used 
    to non-deterministically construct
    all \pn{} words within the 
    \factorWEqnce{} class of a word.
%	(for a proof see Subsection~\ref{subsec:iwm} in the Appendix).
    (A proof that the \pn{} form is 
    indeed obtained by this construction
    can be found in Subsection~\ref{subsec:iwm}.)
    
\end{remark}

%% file: proof_propBinByWeight.tex
\begin{proof}
    We construct a \sumwm{} $\mu$ over 
    the binary alphabet $\Sigma=\{\tzero,\tone\}$.
    Let $\mu(\tone)=2$ and $\mu(\tzero)=1$.
    Then $|w|_\tone + |w| = \mu(w)$ holds
    for any binary word $w\in\Sigma^*$.
    We have
    $f_w(i) + i = \max(|\Fact_i(w)|_\tone) + i
     = \max(\mu(\Fact_i(w))) = f_{w,\mu}(i)$
    and $p_w(i) + i = p_{w,\mu}(i)$ 
    for all $i \in [|w|]$.
    Therefore, $w$ is $\mu$-\pn{} 
    if and only if it is \pn{}.
	\qed
\end{proof}

%% file: proof_theWPNF.tex
\begin{proof}
    Let $\mu$ be a \wm{} over $\Sigma$ w.r.t. the monoid $A$.
    For the first equivalence consider that $\mu$ is gapful.
    Then there exists some word $w\in\Sigma^*$ and an index $i\in[|w|]$
    such that $w$ has a gap at $i$.
    Thus there exists no $n\in A$ for which 
    $f_{w,\mu}(i)=f_{w,\mu}(i-1) \circ_A n$ holds.
    Now suppose there exists some word $w'\in\mathcal{P}_\mu(w)$.
    For such a word $p_{w',\mu}(i)=f_{w',\mu}(i)=f_{w,\mu}(i)$ 
    and $p_{w',\mu}(i-1)=f_{w',\mu}(i-1)=f_{w,\mu}(i-1)$ both must hold.
    Thus we get $f_{w,\mu}(i-1)\circ_A \mu(w'[i])= 
    f_{w',\mu}(i-1)\circ_A \mu(w'[i]) = 
    p_{w',\mu}(i-1)\circ_A\mu(w'[i])=
    p_{w',\mu}(i)=f_{w',\mu}(i)=f_{w,\mu}(i)$.
    Which is a contradiction to the gap,
    so $\mathcal{P}_\mu(w)=\emptyset$ holds.
    For the second direction choose $w\in\Sigma^*$
    with $\mathcal{P}_\mu(w) = \emptyset$.
    Suppose $\mu$ is gapfree,
    so $f_{w,\mu}(i)=f_{w,\mu}(i-1) \circ_A \mu(\ta)$
    holds for all $i\in[|w|]$ and appropriate $\ta\in\Sigma$.
    Then we have a contradiction by constructing 
    a word $w'\in\mathcal{P}_\mu(w)$ as follows:
    Choose $w'[1]\in\Sigma$ with $\mu(w'[1]) = f_{w,\mu}(1)$,
    which is possible according to the assumption for $i=1$.
    And for $i\in[|w|]$ we can inductively choose $w'[i]\in\Sigma$ 
    with $f_{w,\mu}(i)= f_{w,\mu}(i-1) \circ_A \mu(w'[i])$,
    which is also possible according to the assumption.
    Now $p_{w',\mu}=f_{w',\mu}$ and $p_{w',\mu}=f_{w,\mu}$ 
    hold by construction, so $w'\in\mathcal{P}_\mu(w)$ holds.
    
    For the second claim let $\mu$ be not injective. 
    Therefore, we have some distinct letters $\ta,\tb\in\Sigma$
    which have the same weight $\mu(\ta)=\mu(\tb) = i \in A$.
    So $[\ta]_{\sim_\mu} = [\tb]_{\sim_\mu}$ and 
    $p_{\ta,\mu}=f_{\tb,\mu}$ both hold directly.
    Consequently $\{\ta,\tb\} \subseteq \mathcal{P}_\mu(\ta)$ follows.
    For the second direction consider $w,u,v\in\Sigma^*$
    with $u\neq v$ and $\{u,v\} \subseteq \mathcal{P}_\mu(w)$.
    By the definition of $\mathcal{P}_\mu(w)$, the
    \prefWf{} of $u$ and $v$ are both equal to the \factWf{} of $w$.
    So $p_{u,\mu} = p_{v,\mu}$ holds, and therefore $\mu(u[j])=\mu(v[j])$
    holds for all $j\in[|u|]$.
    On the other hand because $u$ and $v$ are 
    different words, there exists some $i\in[|u|]$ with
    $u[i] \neq v[i]$.
    In other words, $\mu$ is not injective.
    
    The third claim follows directly from the first two.
    \qed
\end{proof}

%% file: sec-gwm.tex
%!TEX root = wpnw.tex 
In this section we investigate the behaviour of gapfree \wm s in more 
detail. In order to present a natural and gapfree \emph{standard \wm} 
for ordered alphabets that is \emph{equivalent} to every other 
injective, alphabetically ordered, and gapfree \wm{}
(over arbitrary monoids) - and thus works as a representative, we give 
an alternative condition for gapfree \wm{}s;
the so called \wm s with \emph{stepped based weights}.

First of all, by their definition we can infer that every binary \wm{} 
is gapfree. Consequently we consider non-binary \wm s for the rest of 
this section.
%and \pwm s are in general gapful.

\begin{lemma}
    \label{lem:binwm}
    All binary \wm s are gapfree.    
\end{lemma}
\ifpaper % ONLY PAPER
\else    % ONLY ARCHIVE
\input{proof_lemBinWms}
\fi

\begin{remark}
    By Lemma~\ref{lem:binwm}, we see that when modelling binary \pnity{} 
    by means of weighted \pnity{} (e.g. in the proof of  
    Theorem~\ref{the:weight2bin}) we automatically have the existence of 
    a unique binary \pn{} form as expected.
\end{remark}

In the last section we saw that we have exactly one weighted \pn{} form 
in a \factorWEqnce{} class iff the \wm{} is injective and gapfree.
We now give an alternative condition under which a \wm{} is gapfree, 
which in most cases is easier to check.
Later we will also see that this condition is part of a proper 
characterisation for gapfree \wm s.

\begin{definition}
    \label{def:sbw}
    Let $A$ be a strictly totally ordered monoid.
    A \emph{step function} is a right action of an element $s\in A$ (the 
    \emph{step}) on $A$, i.e. $\sigma_s:A\to A; a\mapsto a\circ_A s$. 
    The \wm{} $\mu$ over $\Sigma$ w.r.t the monoid $A$ is said to have \emph{stepped 
    base weights} if there exists a step function $\sigma_s$ for some $s\in A$ 
    such that
    $\mu(\Sigma)=\{ \sigma_s^i(\min(\mu(\Sigma))) \mid
    i\in[0,|\mu(\Sigma)|-1]\}$ holds.
\end{definition}

In the previous example for $\Sigma =\{\ta,\tb,\tc\}$, the gapfree \sumwm{} 
$\mu$ over $\Sigma$ with $\mu(\ta) = 2$, $\mu(\tb) = 4$,
and $\mu(\tc) = 6$ has stepped base weights with the step of $2$.
In contrast, the gapful \sumwm{} $\nu$ over $\Sigma$ with 
$\nu(\ta) = 1$, $\nu(\tb) = 3$, and $\nu(\tc) = 4$ does not,
because $\nu(\tb) - \nu(\ta) = 2$ but $\nu(\tc) - \nu(\tb)=1$.
In general, stepped base weights imply \emph{gapfreeness} but not vice 
versa (see Subsection~\ref{subsec:fancy} in the appendix).

\begin{proposition}
	\label{prop:step1}
	All \wm s with stepped base weights are gapfree.
\end{proposition}
\ifpaper % ONLY PAPER
\else    % ONLY ARCHIVE
\input{proof_propSTEP1}
\fi

For further investigations of gapfree \wm s we define an equivalence on 
\wm s based on their behaviour on words of the same length.

\begin{definition}\label{defeq}
	Let $\mu_A$ and $\mu_B$ be \wm s over the same alphabet $\Sigma$ 
	w.r.t. the 	monoids $A$ and $B$.
	We say that $\mu_A$ and $\mu_B$ are \emph{equivalent} if for all 
	words $v,w\in\Sigma^n$, for some $n\in\N$, we have 
	$\mu_A(v)<_A\mu_A(w)$ iff $\mu_B(v)<_B\mu_B(w)$.	
\end{definition}

The reasoning behind such an equivalence of \wm s lies in the fact that 
using different but equivalent \wm s does not change their relative 
behaviour. 
Most notably, Definition~\ref{defeq} and the totality of the orders 
imply $\mu_A(v)=\mu_A(w)$ iff $\mu_B(v)=\mu_B(w)$ and therefore, the 
\pn{} form remains. \looseness=-1

For instance, considering again the alphabet $\Sigma =\{\ta,\tb,\tc\}$ 
and the gapfree \sumwm{} $\mu$ over $\Sigma$ with $\mu(\ta) = 2$, 
$\mu(\tb) = 4$, and $\mu(\tc) = 6$ as well as the \prodwm{} $\nu$ over 
$\Sigma$ with $\nu(\ta) = 2$, $\nu(\tb)=6$, and $\nu(\tc)=18$.
Then $\mu$ and $\nu$ are equivalent since they both are alphabetically 
ordered and $2+3^{\frac{\mu(w)}{2}-1}=\nu(w)$ holds for all 
$w\in\Sigma^*$. Therefore, since $\mu$ is gapfree so is $\nu$,
and for instance $\mathcal{P}_\mu(\mathtt{bcac})=\{\mathtt{cbbb}\} = 
\mathcal{P}_\nu(\mathtt{bcac})$ holds.

\begin{proposition}
	\label{prop:EQWMEQPNF}
	The \pn{} form of any word is the same w.r.t. equivalent \wm s, i.e.
	$\mathcal{P}_\mu(w)=\mathcal{P}_\nu(w)$ holds for all $w\in\Sigma^*$ if 
	$\mu$ and $\nu$ are equivalent 
	\wm s.	
\end{proposition}
\ifpaper % ONLY PAPER
\else    % ONLY ARCHIVE
\input{proof_propEQWMPNF}
\fi

Before we present the generalised \wm{}, we prove three auxiliary lemmata
and give the definition of the standard weight measure.

\begin{lemma}
	\label{prop:EQWMS}
	For any two equivalent \wm s, if one of them is gapfree, injective, or 
	alphabetically ordered then so is the other.
%	Let $\mu$ and $\nu$ be equivalent \wm s over $\Sigma$. If $\mu$ is gapfree 
%	then also $\nu$ is gapfree. If $\mu$ is injective
%	then also $\nu$ is injective. Also If $\mu$ is alphabetically ordered
%	then also $\nu$ is alphabetically ordered.
\end{lemma} 
\ifpaper % ONLY PAPER
\else    % ONLY ARCHIVE
\input{proof_propEQWMS}
\fi

Finally, we define the \emph{standard \wm} as an injective gapfree \wm{} 
that is innate to any strictly totally ordered alphabet.

\begin{definition}
	Let $\Sigma=\{\ta_1,\ta_2,\dots,\ta_n\}$ be a strictly totally 
	ordered alphabet, where $n\in\N$. We define the \emph{standard 
	\wm{}} $\mu_\Sigma$ as the alphabetically ordered \sumwm{} over 
	$\Sigma$ with base weights $\mu_\Sigma(\ta_i) = i$ for all $i\in[n]$.
\end{definition}

For instance, considering again the alphabet $\Sigma =\{\ta,\tb,\tc\}$ 
with the usual order, the standard \wm{} $\mu_\Sigma$ has the 
base weights $\mu_\Sigma(\ta)=1$, $\mu_\Sigma(\ta)=2$, and 
$\mu_\Sigma(\ta)=3$. And in the following, we will see that indeed 
$\mu_\Sigma$ is equivalent to both $\mu$ and $\nu$ from the previous 
example.

\begin{lemma}
	\label{lem:STDWM}
	The standard \wm{} is gapfree, injective, and alphabetically ordered.
\end{lemma}
\ifpaper % ONLY PAPER
\else    % ONLY ARCHIVE
\input{proof_lemSTDWM}
\fi

The definition of the equivalence on \wm s raises the question whether 
the standard weight measure is suitable as a representative for all 
gapfree, injective, and alphabetically ordered weight measures. If there 
were other equivalence classes of such \wm s then the standard \wm{} 
would merely represent \emph{one} of many choices. To answer this 
question we first present a peculiar property every gapfree \wm{} has 
and then present our main theorem on the equivalence class of 
the standard \wm.

\begin{lemma}
	\label{lem:cacbtok2}
	Let $\mu$ be an injective and alphabetically ordered \wm{} over $\Sigma$ 
	w.r.t. the monoid $A$.
	Let $\Sigma$ be strictly totally ordered by $<_\Sigma$ and let $\Sigma = 
	\{\ta_1,\dots,\ta_n\}$ with $n\in\N_{>2}$ and $\ta_1 <_\Sigma 
	\ta_2 <_\Sigma \dots <_\Sigma \ta_n$.
	If $\mu$ has no gap over any word of the form $\tc\ta\tc\tb$ where 
	$\ta<_\Sigma\tb<_\Sigma\tc\in\Sigma$ then 
	$\mu(\ta_i\ta_{i+x})=\mu(\ta_{i+y}\ta_{i+x-y})$ holds for all $i,x,y\in\N$ 
	with $y<x$ and $i+x\leq n$.
\end{lemma}
\ifpaper % ONLY PAPER
\else    % ONLY ARCHIVE
\input{proof_lemcacbtok2}
\fi

\begin{theorem}
	\label{the:threeEQ}
	Let $\mu$ be a non-binary, injective, and alphabetically ordered \wm{} over 
	the alphabet $\Sigma$ which is strictly ordered by $<_\Sigma$.
	The following statements are equivalent:
	\\
	1. $\mu$ is gapfree.\\
	2. 	$\mu$ has no gap over any word of the form $\tc\ta\tc\tb$ where 
	$\ta<_\Sigma\tb<_\Sigma\tc\in\Sigma$.\\
	3. $\mu$ is equivalent to the standard \wm{} $\mu_\Sigma$.
%	\begin{enumerate}
%		\item
%		$\mu$ is gapfree.
%		\item
%		$\mu$ has no gap over any word of the form $\tc\ta\tc\tb$ where 
%		$\ta<_\Sigma\tb<_\Sigma\tc\in\Sigma$.
%		\item
%		$\mu$ is equivalent to the standard \wm{} $\mu_\Sigma$.
%	\end{enumerate}
\end{theorem}
\ifpaper % ONLY PAPER
\else    % ONLY ARCHIVE
\input{proof_thethreeEQ}
\fi

Notice, with (1. $\Leftrightarrow$ 2.) in the above we know that any gapful 
\wm{} 
\begin{wrapfigure}[9]{R}{.35\linewidth}
	\vspace*{-0.35cm}
	\centering
	\begin{tikzpicture}[->,>=stealth',shorten >=1pt,auto,node distance=0.5cm,
	semithick]
		\node[rectangle,rounded corners=3,align=center,draw=black,minimum width=0.25cm,minimum height=0.25cm,font=\tiny,fill=white]   (n5)  at (0,0)      {5};
		\node[rectangle,rounded corners=3,align=center,draw=black,minimum width=1.25cm,minimum height=0.7cm,node distance=0.475cm,fill=white,font=\ttfamily]   (a5x)  [below of=n5]          {c\,c\,a~~\textcolor{black!60}{b}c};

		\node[rectangle,rounded corners=3,align=center,draw=black,minimum width=0.25cm,minimum height=0.25cm,font=\tiny,fill=white]   (n4)  at (0.925,-0.1)      {4};
		\node[rectangle,rounded corners=3,align=center,draw=black,minimum width=1cm,minimum height=0.55cm,node distance=0.375cm,fill=white,font=\ttfamily]   (a4)  [below of=n4]          {bc~~c\,b};

		\node[rectangle,rounded corners=3,align=center,draw=black,minimum width=1.25cm,minimum height=0.7cm,node distance=0.475cm,fill=white,opacity=0.5,font=\ttfamily]   (a5)  [below of=n5]          {c\,c\,a~~bc};
		\node[rectangle,rounded corners=3,align=center,draw=black,minimum width=1.25cm,minimum height=0.7cm,node distance=0.475cm,font=\ttfamily]   (a5)  [below of=n5]          {c\,c\,a~~\textcolor{black!60}{b}c};

		\path (a4) edge [out=330,in=290,looseness=1.6,dashed] node [] {}   (a5);
	\end{tikzpicture}
	\vspace*{-0.5cm}

    \caption{Visualisation of the \factWf{}'s gap for $w = \mathtt{ccabccb}$.
        \label{fig:factwei}}
\end{wrapfigure}

\noindent over $\Sigma=\{\ta,\tb,\tc\}$ already has a gap over $\mathtt{bcac}$. For 
instance, consider the sum weight measure $\mu$ over $\Sigma$ with $\mu(\ta)=1$, 
$\mu(\tb)=2$, and $\mu(\tc)=4$.  We see that $\mu$ is gapful, since it has a gap over the word 
$w=\mathtt{ccabccb}$ at index $5$, witnessed by the \factWf{} 
$f_{w,\mu}=(4,8,10,12,15,19,21)$ and the fact that there is no letter with 
weight $15-12=3$. We visualise this gap in Figure~\ref{fig:factwei}.
However, we already have a gap within the even shorter word 
$\mathtt{bcac}$ at index $3$, witnessed by $f_{\mathtt{bcac},\mu}=(4,6,9,11)$.

On the other hand, with (1. $\Leftrightarrow$ 3.) in Theorem~\ref{the:threeEQ} 
we immediately see there only exists \emph{one} equivalence class of gapfree, 
injective, and alphabetically ordered \wm s w.r.t. the same alphabet, 
justifying our choice of $\mu_\Sigma$ as the standard \wm.
Also, since by transitivity all gapfree, injective, and alphabetically 
ordered \wm s w.r.t. to the same alphabet are equivalent, they therefore yield 
the same \pn{} form (by Proposition \ref{prop:EQWMEQPNF}).
In other words, assuming a strictly totally ordered alphabet, every 
word has exactly one weighted \pn{} form that is independent of any chosen 
gapfree, injective, and alphabetically ordered \wm{} over the same alphabet.
With that in mind, paralleling the work presented by Fici and Lipták 
in~\cite{Fici2011}, we introduce \emph{the weighted \pn{} form} of a word 
$w\in\Sigma^*$.

\begin{definition}
	Let $\Sigma$ be a strictly totally ordered alphabet and let $w\in\Sigma^*$.
	We say the $\mu_\Sigma$-\pn{form} is \emph{the weighted \pn{} form} of $w$ 
	or simply the \emph{\pn{} form} of $w$.
\end{definition}

For instance, consider the strictly totally ordered alphabet 
$\Sigma=\{\ta, \tb, \tc\}$, with the standard \wm{} $\mu_\Sigma$ such that  
$\mu_\Sigma(\ta)=1,\mu_\Sigma(\tb)=2$, $\mu_\Sigma(\tc)=3$. The weighted \pn{} 
form of $\mathtt{bcac}$ is $\mathtt{cbbb}$, since 
$\mathcal{P}_{\mu_\Sigma}(\mathtt{bcac})=\{\mathtt{cbbb}\}$ holds as seen 
in previous examples.
With Theorem~\ref{the:threeEQ} the same also holds for any other gapfree, 
injective, and alphabetically ordered \wm{}.

\begin{remark}	
	By Theorem~\ref{the:threeEQ} we immediately see that the gapfree 
	property of a weight measure is decidable. Since any gapful \wm{} 
	already has a gap over a word of length four using three letters,
	one can check whether a \wm{} is gapfree in the following way:
	test for all $\binom{|\Sigma|}{3}$ possible 
	enumerations of three letters $\ta<_{\Sigma}\tb<_{\Sigma}\tc$ whether there 
	exist an $\tx\in\Sigma$ with $\mu(\tb\tx) = \mu(\ta\tc)$. 
	Notice, that $\Sigma$ is finite and we obtain a running time of 
	$\mathcal{O}(|\Sigma|^4)$.
\end{remark}

%% file: proof_lemBinWms.tex
\begin{proof}
    Let $\mu$ be a binary \wm{} over $\Sigma$ w.r.t. $A$ and with the two base weights 
    $\mu(\Sigma)=\{x,y\}$, where $x<_A y$.
    W.l.o.g. let $\mu$ be injective, so $\Sigma$ is binary as well.
    Furthermore w.l.o.g. let $\Sigma = \{\tzero,\tone\}$ and 
    $\mu(\tzero)=x$, $\mu(\tone)=y$.
    Now let $w\in\Sigma^*$ and $i\in[|w|]$.
    Then $f_{w,\mu}(i)$ is realised by some factor
    $u\in\Fact_i(w)$ with $\mu(u)=f_{w,\mu}(i)$ and
    $f_{w,\mu}(i-1)$ is realised by some factor
    $v\in\Fact_i(v)$ with $\mu(v)=f_{w,\mu}(i-1)$.
    Now $|v|_\tone-|u|_\tone \in \{0,1\}$ holds because
    otherwise $\mu(v)$ or $\mu(u)$ would not 
    be the maximum weight a factor of length $i$ or $i-1$ has.
    In total either $f_{w,\mu}(i)=f_{w,\mu}(i-1) \circ_A \mu(\tone)$
    or $f_{w,\mu}(i)=f_{w,\mu}(i-1) \circ_A \mu(\tzero)$ holds.
    Therefore $\mu$ is gapfree.
    \qed
\end{proof}

%% file: proof_propSTEP1.tex
\begin{proof}
    Let $\mu$ be a non-binary \wm{} over $\Sigma$ w.r.t $A$ with stepped base 
    weights. W.l.o.g let $\mu$ be injective and let $\Sigma$ be of the form 
    $\{\ta_0, \ta_1, \dots , \ta_{n-1}\}$, where the letters are in ascending 
    order of their weight, so $\mu(\ta_i)<_A\mu(\ta_{i+1})$ holds for all 
    $i\in[0,n-2]$.
    
    Assume there exists a step function $\sigma$
    with the step $s\in A$ such that
    $\mu(\Sigma)$ is of the form 
    $\{ \sigma^i(\min(\mu(\Sigma))) \mid i\in[0,|\mu(\Sigma)|-1]\}$.
    Consequently the weight of every letter in $\Sigma$ is 
    $\mu(\ta_i) = \sigma^{i}(\min(\mu(\Sigma)))$
    for all $i\in[0,n-1]$.
    In particular we have $\mu(\ta_0)=\min(\mu(\Sigma))$.
    Now consider some word $w\in\Sigma^*$ and index $l\in[|w|]$,
    then $f_{w,\mu}(l)$ is realised by 
    some factor $\ta_{p_1}\dots\ta_{p_l}\in\Fact_l(w)$
    with some sequence $p_1,\dots,p_l\in[0,n-1]$.
    Therefore $f_{w,\mu}(l)$ is of the form 
    $\sigma^{p_1}(\mu(\ta_0))\circ_A\dots\circ_A\sigma^{p_l}(\mu(\ta_0))$.
    And similarly $f_{w,\mu}(l-1)$ is realised by 
    some other factor $\ta_{q_1}\dots\ta_{q_{l-1}}\in\Fact_{l-1}(w)$
    for some sequence $q_1,\dots,q_{l-1}\in[0,n-1]$, 
    and we have $f_{w,\mu}(l-1) =
    \sigma^{q_1}(\mu(\ta_0))\circ_A\dots\circ_A\sigma^{l_i}(\mu(\ta_0))$.
    Now let $m = \sum_{i=1}^{l}(p_i)$ and $o = \sum_{i=1}^{l-1}(q_i)$
    be the number of steps in the weights
    $f_{w,\mu}(l)$ and  $f_{w,\mu}(l-1)$.
    So they are the maximum number of 
    steps any of $w$'s factors of length $l$
    and $l-1$ can have in their weight.
    
    First of all $m-o \geq 0$ holds because we 
    know the \factWf{} is strictly increasing and 
    otherwise $m$ would not be the maximum for length $l$.
    We also know $m-o < n$ holds because
    if $m-o\geq n$ held, $o$ would not be 
    the maximum number of steps in a factor of length $l-1$.
    We now know 
    $f_{w,\mu}(l)$ and $f_{w,\mu}(l-1)$ only differ by
    $k := m-o$ steps where $0\leq k < n$ holds.
    Consequently $f_{w,\mu}(l) = f_{w,\mu}(l-1) \circ_A \sigma^k(\ta_0)$
    holds and because $\mu$ has stepped base 
    weights there exists such a letter 
    $\ta_k\in\Sigma$ with $\mu(\ta_k) = \sigma^k(\ta_0)$.
    Thus $\mu$ is gapfree.
    \qed
\end{proof}

%% file: proof_propEQWMPNF.tex
\begin{proof}
	Assume $\mu$ and $\nu$ are equivalent \wm s over $\Sigma$.
	Then for all words $w\in\Sigma^*$ and $u\in\Fact_{i}(w)$ with $i\in[|w|]$
	it holds $\mu(u)=f_{w,\mu}$ iff $\nu(u)=f_{w,\nu}$. Consequently with the 
	construction given in Remark~\ref{rem:pnfc} we have that the $\mu$-\pn{} 
	form is the same as the $\nu$-\pn{} form.\qed
\end{proof}

%% file: proof_propEQWMS.tex
\begin{proof}
	Let $\mu_A,\mu_B$ be equivalent \wm s over $\Sigma$ w.r.t. monoids 
	$A,B$.
	
	1) Assume $\mu_A$ gapfree but suppose $\mu_B$ not. There exists a word 
	$w\in\Sigma^*$ with a gap at $i\in[|w|]$ regarding $\mu_B$.
	So for every $\tx\in\Sigma$ we have 
	$f_{w,\mu_B}(i)=\mu_B(u)\neq\mu_B(v\tx)=f_{w,\mu_B}(i-1)\circ_B\mu_B(\tx)$ 
	where $u\in\Fact_i(w)$ and $v\in\Fact_{i-1}(w)$. Since $\mu_A$ and $\mu_B$ 
	are equivalent also $\mu_A(u)\neq\mu_A(v\tx)$ holds for all $\tx\in\Sigma$. 
	This is a contradiction since $\mu_A$ is gapfree.
	
	2) Assume $\mu_A$ injective but suppose $\mu_B$ not. There exist letters 
	$\ta,\tb\in\Sigma$ with $\mu_B(\ta)=\mu_B(\tb)$. Since $\mu_A$ and $\mu_B$ 
	are equivalent also $\mu_A(\ta)=\mu_A(\tb)$ holds contradicting the 
	assumption.
	
	3) Follows directly by the definition of equivalent \wm s.
	\qed
\end{proof}

%% file: proof_lemSTDWM.tex
\begin{proof}
	The standard \wm{} is gapfree by Proposition~\ref{prop:step1}, since it is 
	a \sumwm{} with stepped base weights. It is injective and alphabetically 
	ordered by definition.
\end{proof}

%% file: proof_lemcacbtok2.tex
\begin{proof}
	By Induction over $x$. The case for $x=1$ is trivial.
	
	Firstly consider the case $x=2$. 
	W.l.o.g. let $i=1$, so in this case we show that 
	$\mu(\ta_1\ta_3)=\mu(\ta_2\ta_2)$ holds.
	Consider $u=\ta_3\ta_1\ta_3\ta_2$.
	Assuming $\mu$ has no gaps over words of this form and since 
	$f_{u,\mu}(3)=\mu(\ta_3\ta_1\ta_3)$ and $f_{u,\mu}(2)=\mu(\ta_3\ta_2)$ hold 
	we know there exists some $z\in[n]$ such that 
	$\mu(\ta_1\ta_3)=\mu(\ta_2\ta_z)$.
	Since $\mu(\ta_i)<_A\mu(\ta_{i+1})$ we have
	$\mu(\ta_1)<_A\mu(\ta_z)<_A\mu(\ta_3)$.
	Therefore $z=2$ and $\mu(\ta_1\ta_3) = \mu(\ta_2\ta_2)$ hold.
	
	Secondly consider $x>2$. Let w.l.o.g $y\leq\lfloor\frac{x}{2}\rfloor$.
	By induction we assume the claim holds for all smaller $x$, e.g.
	\begin{inparaenum}[(1)]
		\item $\mu(a_{i}a_{i+x-1})=\mu(a_{i+y}a_{i+x-y-1})$ holds and also
		\item 
		$\mu(a_{j}a_{j+y+1})=\mu(a_{j+1}a_{j+y})$ 
		where $j=i+x-y-1$ holds since 
	$y+1\leq\lfloor\frac{x}{2}\rfloor+1<x$.
	\end{inparaenum}
	By (1) we have
	$\mu(a_{i}a_{i+x-1}a_{i+x})=\mu(a_{i+y}a_{i+x-1-y}a_{i+x})$ and by 
	(2) rewritten as $\mu(a_{i+x-y-1}a_{i+x})=\mu(a_{i+x-y}a_{i+x-1})$
	we get $\mu(a_{i+y}a_{i+x-1-y}a_{i+x}) = 
	\mu(a_{i+y}a_{i+x-y}a_{i+x-1})$. Therefore 
	$\mu(a_{i}a_{i+x})=\mu(a_{i+y}a_{i+x-y})$ holds.
	\qed
\end{proof}

%% file: proof_thethreeEQ.tex
\begin{proof}
	(1. $\Rightarrow$ 2.) Follows immediately from the definition of gapfree 
	\wm s.
	
	(2. $\Rightarrow$ 3.) Let %$\Sigma$ like in Lemma~\ref{lem:cacbtok2}, i.e. 
	$\Sigma = \{a_1,\dots,a_n\}$ with $n\in\N_{>2}$ and $\ta_1 <_\Sigma \ta_2 
	<_\Sigma \dots <_\Sigma \ta_n$.
	Let $k\in\N$ and $v=\ta_{i_1}\dots\ta_{i_k}\in\Sigma^k$ and 
	$w=\ta_{j_1}\dots\ta_{j_k}\in\Sigma^k$ for all $i_\ell,j_\ell\in[n]$ and 
	$\ell\in[k]$.
	Now $\mu_\Sigma(v)=\sum_{\ell=1}^{k}i_\ell$ and 
	$\mu_\Sigma(w)=\sum_{\ell=1}^{k}j_\ell$ 
	hold.
	
	We show $\mu(v) <_A \mu(w) \Leftrightarrow \mu_\Sigma(v) < \mu_\Sigma(w)$ 
	holds by induction over $k$:
	
	Case $k=1$: Trivial since $\mu$ and $\mu_\Sigma$ are alphabetically ordered.
	
	For the further cases assume w.l.o.g. that $v$ and $w$ share no letters and 
	let their letters be ordered increasingly, i.e. let $i_\ell\neq j_{\ell'}$ 
	and 
	$i_\ell\leq i_{\ell+1}$ and $j_l\leq j_{\ell+1}$ for all 
	$\ell,\ell'\in[k-1]$, furthermore 
	w.l.o.g. let $i_1<j_1$.
	
	Case $k>2$: We consider two subcases dependent on $i_k$ and $j_1$. Notice 
	$j_1=i_k$ can not occur since $v$ and $w$ share no letters.
	
	Subcase $i_k<j_1$: In this case we know $i_\ell<j_{\ell'}$ for all 
	$\ell,\ell'\in[k]$. 
	Therefore $\mu(v)<_A\mu(w)$ and $\mu_\Sigma(v)<_A\mu_\Sigma(w)$ both hold 
	immediately.
	
	Subcase $j_1<i_k$: In this case we choose $x=i_k-i_1$ and $y=j_1-i_1$, 
	consequently $y<x$ holds. By Lemma~\ref{lem:cacbtok2} we have 
	$\mu(v)=\mu(\ta_{i_1}\dots\ta_{i_{k-1}}\ta_{i_1+x})=
	\mu(\ta_{i_1+y}\ta_{i_2}\dots\ta_{i_{k-1}}\ta_{i_1+x-y})$. 
	The claim follows since $i_1+y=j_1$ and by the induction hypotheses
	we have $\mu(v')<_A\mu(w') \Leftrightarrow \mu_\Sigma(v')<\mu_\Sigma(w')$ 
	for 
	$v'=v[2\dots k]$ and $w'=w[2\dots k]$.
	
	(3. $\Rightarrow$ 1.) Follows immediately by Proposition~\ref{prop:EQWMS} 
	since the standard \wm{} is gapfree.
	\qed
\end{proof}

%% file: conclusion.tex
%!TEX root = wpnw.tex 
In this work we presented the generalisation of \pnity{} 
on binary alphabets as introduced by \cite{Fici2011}
to arbitrary alphabets by applying weights to the letters 
and comparing the weight of a factor with the weight of the prefix 
of the same length.\looseness=-1

Since one of the main properties of binary \pnity{},
namely the existence of a unique \pn{} form, does 
not hold for weighted \pnity{} with arbitrary \wm s,
we investigated necessary restrictions to obtain 
a unique \pn{} form even in the generalised setting.
Here, it is worth noticing that we did not only generalise
the size of the alphabet but also the weights are rather general:
they belong to any (totally ordered) monoid.
This is of interest because some peculiarities do not occur
if $\N$ or $\N_0$ are chosen. \looseness=-1
In Section~\ref{sec:wpn} we proved that there always exists 
a unique \pn{} form if the \wm{} is gapfree and injective.
In Section~\ref{sec:gwm} we further demonstrated that all gapfree \wm s over 
the same alphabet are equivalent and therefore every word has the same weighted 
\pn{} form w.r.t. each of them.
Which led to the definition of the standard \wm{} and ultimately to a unique
\pn{} form in the generalised setting that exists independent of chosen \wm s.
Additionally, we showed that \emph{gapfreeness} as a property of \wm s is 
decidable and can easily be checked in time $\mathcal{O}(|\Sigma|^4)$.

However, the exact behaviour of the weighted \pn{} form,
or generally \factorWEqnt{} words,
especially regarding changes in their Parikh vectors,
remains an open problem.
Moreover, a reconnection of weighted \pnity{} to the initial problem 
of indexed jumbled pattern matching would be of some interest and might
prove useful when investigating pattern matching problems w.r.t.
a non-binary alphabet.

Finally, we like to mention that an easier, but weaker, approach to work with 
prefix normality on arbitrary alphabets can be achieved by considering a subset 
$X$ of $\Sigma$: each letter $\ta$ in a word is treated like a $\tone$ if 
$\ta\in X$ and $\tzero$ otherwise, which can also be expressed by weighted 
\pnity{} (see Subsection~\ref{subsec:subset} in the Appendix). \looseness=-1

%% file: furtherinsights.tex
In this section we present results which are not necessarily important to 
understand the weighted \pnity{} but which give a more detailed insight,
e.g. in the relation of weighted prefix normality and the original 
\pnity{}, introduced in \cite{Fici2011}.
Therefore, we start this section with the adaption of the \emph{position 
function} into the weighted setting and present the analogous results. 
Afterwards, we prove that the converse of Proposition~\ref{prop:step1} does not 
hold true in general, i.e. not every gapfree \wm{} has stepped based weights. 
Following that, we present some insights on how to obtain an injective \wm{} 
from a non-injective one. We end this part with the na\"ive alternative 
approach to generalise the binary \pnity, namely by \emph{subset \pnity}, and 
prove that this generalisations is already covered by weighted \pnity.

%%%%%%%%%%%%%%%%%%%%%%%%%%%%%%%%%%%%%%%%%%
\subsection{Weighted Position Functions}
\label{subsec:wposf}
%%%%%%%%%%%%%%%%%%%%%%%%%%%%%%%%%%%%%%%%%%
In the following we define a more general version of 
the binary position function defined in~\cite{Fici2011}.
With the binary context in mind this function is defined to 
give the position of the \nth{$i$} $\tone$ in the word $w$,
i.e. $\pos_w(i) := \min\{k \mid p_w(k)=i\}$
for all $i\in[p_w(|w|)]$ and $w\in\{\tzero,\tone\}^*$.
However, in the weighted context for most words not 
every weight corresponds to a prefix with exactly the same weight.
Thus, we do not define a single exact position function,
but two functions which together 
enclose the position within a word where a certain weight is reached.
Only if both functions return the same position for some word and weight,
that word's prefix up to that position has exactly that weight.

\begin{definition}
    \label{def:posWf}
    Let $w\in\Sigma^*$ and let $\mu$ be a \wm{} 
    over $\Sigma$ w.r.t. the monoid $A$.
    We define the \emph{\maxposWf{}}  and \emph{\minposWf{}} respectively by
    $\maxpos_{w,\mu} : A \to [|w|]_0,\,
    i \mapsto \max \{k \in [|w|]_0 \mid p_{w,\mu}(k) \leq i\}$ and
    $\minpos_{w,\mu} :  A \to [|w|]_0,\,
    i \mapsto \min \{k \in [|w|]_0 \mid p_{w,\mu}(k) \geq i\}$.
\end{definition}

By this definition, we are able to prove similar statements to the binary prefix 
normality.

\noindent
% \textbf{Proof of Lemma~\ref{lem:basics}.}
\input{lembasics}
\newpage
\input{proof_lemBasics}

\noindent
% \textbf{Proof of Lemma~\ref{lem:f1}.}
\input{lemf1}
\input{proof_lemF1}

\noindent
% \textbf{Proof of Proposition~\ref{prop:eqivPN}.}
\input{propequivpn}
\input{proof_propEqivPN}

%%%%%%%%%%%%%%%%%%%%%%%%%%%%%%%%%%%%%%%%%%
\subsection{Prime \WM s}
\label{subsec:pwms}
%%%%%%%%%%%%%%%%%%%%%%%%%%%%%%%%%%%%%%%%%%
In this subsection we briefly examine \pwm s regarding a certain unique 
properties they inherit from the prime numbers. Notice, for injective \pwm{} 
the weight of any word is characteristic 
for its Parikh vector by the uniqueness of 
the prime number factorisation: any two words 
with the same weight under an injective 
\pwm{} must have exactly the same letters, 
i.e. the same Parikh vectors. 
Thus, it is not unreasonable to assume they might prove interesting 
regarding the \emph{Indexed JPM} (IJPM), since in the binary case \pn{} forms 
always have the same Parikh vectors as their equivalent words.
In our more general case, we saw that for some words the weighted \pn{} 
form has a different Parikh vector.
However, this is prohibited by the above mentioned property of \pwm s.
Consequently, every \pwm{} has gaps over these words and is therefore gapful.

\begin{lemma}
	\label{lem:primeWM}
	All non-binary \pwm s are gapful.
\end{lemma}
\input{proof_lemPrimeWM}

In some sense the \pwm{} even is the \emph{most} gapful \wm{}, since it has 
gaps between all of its base weights by the definition of prime numbers. So if 
there exists a \wm{} that has a gap over some word, then also every \pwm{} has 
a gap over that word. This sentiment leads us to believe that \pwm s might not 
be as helpful in solving the IJMP as initially assumed.

%%%%%%%%%%%%%%%%%%%%%%%%%%%%%%%%%%%%%%%%%%
\subsection{Gapfree and Stepped Based Weights}
\label{subsec:fancy}
%%%%%%%%%%%%%%%%%%%%%%%%%%%%%%%%%%%%%%%%%%

In this subsection, we prove that the converse of Proposition~\ref{prop:step1} 
does not hold in general, i.e. stepped based weights and \emph{gapfreeness} are 
not equivalent. For this purpose, we define a relatively technical monoid $V$
equipped with a weight-function $\mu$.

\begin{definition}
Let $\Sigma=\{\ta,\tb,\tc\}$ and let $V$ be the strictly totally ordered monoid
    where $V=\{\colvec{a\\b} \mid a,b\in\N_0 \}$, $\circ_V$ is the usual 
    addition on vectors, $\neutr_V = \colvec{0\\0}$,
    and $<_V$ is the order obtained by the lexicographical expansion of the 
    usual \emph{less than} onto vectors, e.g $\colvec{0\\0}\prec\colvec{0\\2}
    \prec\colvec{1\\1}\prec\colvec{2\\0}$ holds.
\end{definition}

\begin{lemma}
	The \wm{} $\mu$ over $\Sigma$ w.r.t. $V$ with the base weights
	$\mu(\ta) = \colvec{0\\2}$, $\mu(\tb) = \colvec{1\\1}$, and 
	$\mu(\tc) = \colvec{2\\0}$ does not have stepped base weights.
\end{lemma}
\begin{proof}
	It is easy to see that there exists no $x\in V$ with 
	$\colvec{0\\2}+x=\colvec{1\\1}$, since $1 < 2$. So there exists no step 
	function for $\mu$.\qed
\end{proof}

\begin{lemma}
    \label{applem:fancymonoid}
    
    The \wm{} $\mu$ over $\Sigma$ w.r.t. $V$ with the base weights
    $\mu(\ta) = \colvec{0\\2}$,
    $\mu(\tb) = \colvec{1\\1}$, and 
    $\mu(\tc) = \colvec{2\\0}$ is gapfree.
\end{lemma}
%\textbf{Proof of Proposition~\ref{prop:fancymonoid}.}
\input{fancymonoid}

\subsection{Injective \WM s}
%%%%%%%%%%%%%%%%%%%%%%%%%%%%%%%%%%%%%%%%%%%%%%%%%%%%%%%%%%%%%%%%%%%%%%
\label{subsec:iwm}
\input{sec-iwm}

%%%%%%%%%%%%%%%%%%%%%%%%%%%%%%%%%%%%%%%%%%
\subsection{Subset Prefix Normality}
\label{subsec:subset}
%%%%%%%%%%%%%%%%%%%%%%%%%%%%%%%%%%%%%%%%%%
\input{section-xpnw}

%% file: lembasics.tex
\begin{lemma}
    \label{lem:basics}
    Let $\mu$ be a \wm{} 
    over $\Sigma$ w.r.t. a monoid $A$,
    $w\in\Sigma^*$, $j,k \in [|w|]_0$ 
    and $x,y\in A$.
    Then $p_{w,\mu}$ and $f_{w,\mu}$ have the following properties:
    %    p and f are strictly increasing
    
    \begin{enumerate}[(1)]
        \item
        $j < k$ iff $f_{w,\mu}(j) \prec f_{w,\mu}(k)$ 
        iff $p_{w,\mu}(j) \prec p_{w,\mu}(k)$, 
        
        \item
        $p_{w,\mu}(\maxpos_{w,\mu}(x)) \preceq
        x \preceq p_{w,\mu}(\minpos_{w,\mu}(x))$, 
        
        \item
        $\maxpos_{w,\mu}(p_{w,\mu}(k))
        = \minpos_{w,\mu}(p_{w,\mu}(k)) = k$,
        
        \item
        if $\maxpos_{w,\mu}(x) < j$ then $x \prec p_{w,\mu}(j)$ and 
        if $j < \minpos_{w,\mu}(x)$ then $p_{w,\mu}(j) \prec x$,         
        
        \item
        $\maxpos_{w,\mu}(x) \leq \minpos_{w,\mu}(x)$,
        
        \item
        $x\prec y$ implies $\maxpos_{w,\mu}(x) \leq \maxpos_{w,\mu}(y)$
        as well as $\minpos_{w,\mu}(x) \leq \minpos_{w,\mu}(y)$.
    \end{enumerate}    
\end{lemma}

%% file: proof_lemBasics.tex
\begin{proof}
    \quad
    \begin{enumerate}[(1)]
        
        %1       
        \item 
         With the increasing property of \wm s the equivalences
         follow from the definition of the \factWf{}
         as the maximum over all factors and the fact
         that every prefix itself 
         is a prefix of every longer prefix.
        
        %2
        \item
         Directly follows by the definition of the max-position and \minposWf{}        
         as $\maxpos_{w,\mu}(x) = \max \{i \in [|w|]_0 \mid p_{w,\mu}(i) \preceq x\}$ and $\minpos_{w,\mu}(x) = \min \{i \in [|w|]_0 \mid x \preceq p_{w,\mu}(i)\}$.
        
        %3
        \item Follows by the definition of the max-position and \minposWf{}
        and the fact that the \prefWf{} is strictly increasing (see (1)).
        
        %4
        \item Follows by the definition of the \maxposWf{} as a maximum
        and the \minposWf{} as a minimum.
        
        %5
        \item
        Follows by (1) and (2). 
        
        %6
        \item Suppose otherwise, so let $x\prec y$ 
        but $\maxpos_{w,\mu}(x) > \maxpos_{w,\mu}(y)$ holds.
        With (4), we then have $y \prec p_{w,\mu}(\maxpos_{w,\mu}(x))$ and
        with (2) we have $p_{w,\mu}(\maxpos_{w,\mu}(x)) \preceq x$.
        Together these are a contradiction to $x\prec y$.
        Now suppose $\minpos_{w,\mu}(x) > \minpos_{w,\mu}(y)$ holds.
        Again with (4), 
        we have $p_{w,\mu}(\minpos_{w,\mu}(y)) \prec x$ and
        with (2) we have $y \preceq p_{w,\mu}(\minpos_{w,\mu}(y))$.
        Which is again a contradiction to $x\prec y$.
        
        \qed
    \end{enumerate}
    \renewcommand{\qed}{}
\end{proof}

%% file: lemf1.tex
\begin{lemma}
    \label{lem:f1}
    For a \wm{} $\mu$  
    over the alphabet $\Sigma$ w.r.t. a monoid $(A,\circ)$
    and $w\in\Sigma^*$, 
    $f_{w,\mu}(j) \preceq f_{w,\mu}(i) \circ f_{w,\mu}(j-i)$ holds
    for all $i,j \in [|w|]_0$ with $i < j$.
\end{lemma}

%% file: proof_lemF1.tex
\begin{proof}
    Let $i,j\in[|w|]_0$ be indexes with $i<j$.
    Now suppose $f_{w,\mu}(j) \succ f_{w,\mu}(i) \circ f_{w,\mu}(j-i)$.
    Let $u\in\Fact_j(w)$ be a factor of $w$
    with $\mu(u) = f_{w,\mu}(j)$.
    Then by the definition of the \factWf, 
    $\mu(u[1\dots i]) \preceq f_{w,\mu}(i)$ and 
    $\mu(u[(i+1)\dots j]) \preceq f_{w,\mu}(j-i)$ both hold.
    And thus
    $f_{w,\mu}(j) \succ \mu(u[1\dots i]) \circ \mu(u[(i+1)\dots j]) = 
    \mu(u)$ holds.
    This is a contradiction because $u$ was chosen with $\mu(u) = f_{w,\mu}(j)$, so the original claim follows.
    \qed
\end{proof}

%% file: propequivpn.tex
\begin{lemma}
    \label{prop:eqivPN}
    For a \wm{} $\mu$ over the alphabet $\Sigma$
    w.r.t. a monoid $(A,\circ)$ and $w\in\Sigma^*$
    the following properties are equivalent:
    
    \begin{enumerate}[(1)]
        \item
        $w$ is $\mu$-\pn{},
        
        \item
        $p_{w,\mu}(j) \preceq p_{w,\mu}(i) \circ p_{w,\mu}(j-i)$
        for all $i,j \in [|w|]_0$ with $i < j$,
            
        \item
        $\minpos_{w,\mu}(\mu(v)) \leq |v|$
        for all $v \in \Fact(w)$,
               
        \item
        $\maxpos_{w,\mu}(a) + \minpos_{w,\mu}(b)
        \leq \minpos_{w,\mu}(a \circ b)$
        for all $a,b\in A$, with $a \circ b\preceq\mu(w)$.        
    \end{enumerate}
\end{lemma}

%% file: proof_propEqivPN.tex
\begin{proof}
    \medskip
    (1)$\Rightarrow$(2).
    Follows by the second additional lemma,
    since for any \pn{} word the prefix- and \factWf{} are equal
    by definition.
    \medskip
    
    (2)$\Rightarrow$(3).
    Assume we have (2) but suppose there exists $v \in \Fact(w)$
    with $|v|<\minpos_{w,\mu}(\mu(v))$. 
    Now we write $v$ as $v = w[i+1\dots j]$ for some $i,j\in[|w|]_0$ with $i<j$.
    Then we have $p_{w,\mu}(j)=p_{w,\mu}(i) \circ \mu(v)$.
    And by (8) of the first additional lemma, we know that 
    $\mu(v) \succ p_{w,\mu}(|v|)$ holds so in total we have
    $p_{w,\mu}(j) \succ p_{w,\mu}(i)\,\circ\,p_{w,\mu}(|v|)$.
    Which is a contradiction to (2), because we have
    $|v|=|w[i+1\dots j]|=j-i$.
    \medskip
    
    (3)$\Rightarrow$(1).
    Assume we have (3).
    Let $i\in[|w|]$ and 
    let $v\in\Fact(w)$ with $\mu(v)=f_{w,\mu}(i)$ 
    then we have $|v| \geq \minpos_{w,\mu}(\mu(v))$.
    By (2) and (4) of the first additional lemma follows 
    $p_{w,\mu}(|v|)
    \succeq p_{w,\mu}(\minpos_{w,\mu}(\mu(v)))
    \succeq \mu(v)$ from which (1) follows directly
    because we now have $p_{w,\mu}(i) \succeq f_{w,\mu}(i)$.
    \medskip
    
    (3)$\Rightarrow$(4).
    Let $a,b \in A$ with $a\circ b \preceq \mu(w)$,
    $m = \minpos_{w,\mu}(a \circ b)$ and $n = \maxpos_{w,\mu}(a)$.
    Now consider $w$'s factor $v=w[n+1\dots m]$
    which has a length of $m-n$.
    So $p_{w,\mu}(n) \circ \mu(v) = p_{w,\mu}(m)$
    and $|v| = m-n$ each follow.
    By (3) and (4) of the first additional lemma  we know
    $a \circ \mu(v) \succeq a \circ b$
    and therefore $\mu(v)\succeq b$ holds.
    Again by (11) of the first additional lemma  we get
    $\minpos_{w,\mu}(\mu(v)) \geq \minpos_{w,\mu}(b)$.
    So in total with (3) follows
    $\minpos_{w,\mu}(b) \leq \minpos_{w,\mu}(\mu(v)) \leq |v| =
    m-n=\minpos_{w,\mu}(a\circ b) - \maxpos_{w,\mu}(a)$.
    \medskip
    
    (4)$\Rightarrow$(3).
    Let $v\in\Fact(w)$, we write $v$ as 
    $v = w[i+1\dots j]$ for some $i,j\in[|w|]_0$.
    So with the first additional lemma we have
    $\minpos_{w,\mu}(p_{w,\mu}(i)) = i$ and 
    $\minpos_{w,\mu}(p_{w,\mu}(i) \circ \mu(v)) = \minpos_{w,\mu}(p_{w,\mu}(j)) = j$.
    By (4) we then have
    $\minpos_{w,\mu}(\mu(v)) \leq \minpos_{w,\mu}(p_{w,\mu}(i) \circ \mu(v))
    - \minpos_{w,\mu}(p_{w,\mu}(i)) = j-i = |v|$.
    \qed
\end{proof}

%% file: proof_lemPrimeWM.tex
\begin{proof}
    Since $\mu$ is a non-binary \pwm{},
    we have $\mu(\Sigma)\subseteq\P$ and $|\mu(\Sigma)|>2$.
    Thus, there exist letters  $\ta_1,\ta_2$, and $\ta_3$
    such that $\mu(\ta_1),\mu(\ta_2)$,
    and $\mu(\ta_3)$ are pairwise different prime numbers.
    Assume w.l.o.g. $\mu(\ta_1)<\mu(\ta_2)<\mu(\ta_3)$.
    Suppose that $\mu$ is gapfree, i.e.
    for all words $w\in\Sigma^{\ast}$ and all 
    $i\in[|w|-1]$ there exists an $x\in\Sigma$ with
    $f_{w,\mu}(i+1)=f_{w,\mu}(i)*\mu(x)$.
    Consider $w=\ta_2\ta_3\ta_1\ta_3$.
    By $\mu(\ta_2\ta_3\ta_1)<\mu(\ta_3\ta_1\ta_3)$ and 
    $\mu(\ta_1\ta_3)=\mu(\ta_3\ta_1)<\mu(\ta_2\ta_3)$
    we get $f_{w,\mu}(3)=\mu(\ta_3\ta_1\ta_3)$ and 
    $f_{w,\mu}(2)=\mu(\ta_2\ta_3)$.
    By the supposition, there exists an $x\in\Sigma$ with
    $\mu(\ta_3\ta_1\ta_3) = f_{w,\mu}(3) = 
    f_{w,\mu}(2) * \mu(x) = \mu(\ta_2\ta_3) * \mu(x)$.
    In other words $\mu(\ta_1)*\mu(\ta_3)=\mu(\ta_2)*\mu(x)$
    must hold.
    This is a direct contradiction to the uniqueness of the 
    prime number factorisation.
    Therefore no non-binary \pwm{} can be gapfree,
    as witnessed by the word examined above.
    \qed
\end{proof}

%% file: fancymonoid.tex
\begin{proof}
    For $i\in[|w|]$ let $u\in\Fact_{i+1}(w)$ be the factor determining $f_{w,\mu}(i+1)$
    and $v\in\Fact_{i}(w)$ be the factor determining $f_{w,\mu}(i)$
    such that $i$ is minimal with $u$ and $v$ not overlapping
    (if they overlap, the non-overlapping parts are taken as $u$ and $v$ respectively).
    Now choose $r,s,t\in\Z$ with 
    $r=|u|_{\ta}-|v|_{\ta}$,
    $s=|u|_{\tb}-|v|_{\tb}$,
    and $t=|u|_{\tc}-|v|_{\tc}$.
    Thus we have $r+s+t=1$ by 
    \[
    r + s + t
    = |u|_{\ta}-|v|_{\ta} + |u|_{\tb}-|v|_{\tb} + |u|_{\tc}-|v|_{\tc}
    = |u|-|v|
    = i+1 - i
    = 1
    \text{.}
    \]
    Moreover we have
    \[
    \mu(u)
%    = f_{w,\mu}(i+1)
    = \binom{|u|_{\tb}+2|u|_{\tc}}{2|u|_{\ta}+|u|_{\tb}}
    = \binom{|v|_{\tb}+s+2|v|_{\tc}+2t}{2|v|_{\ta}+2r+|v|_{\tb}+s}
%    = \binom{|v|_{\tb}+2|v|_{\tc}}{2|v|_{\ta}+|v|_{\tb}}
%    + \binom{s+2t}{2r+s}
%    = f_{w,\mu}(i) + \binom{s+2t}{2r+s}
    = \mu(v) + \binom{s+2t}{2r+s}
    \text{.}
    \]
%    To prove $\mu$ is gapfree,
%    we show there exists a letter in
%    $\Sigma$ with a weight of $\binom{s+2t}{2r+s}$.
    And with 
    $
    |v|_{\tb}
     = |u|_{\tb} - s
     = |u|_{\tb} + r + t - 1
    $
    we get
    \[
    \mu(u) = \binom{|u|_{\tb}+2|v|_{\tc}+2t}{|u|_{\tb}+2|v|_{\ta}+2r}
    \text{ and }
    \mu(v)=\binom{|u|_{\tb}+r+t-1+2|v|_{\tc}}
    {2|v|_{\ta}+|u|_{\tb}+r+t-1}
    \text{.}
    \]
    By evaluating $f_{w,\mu}(i+1) = \mu(u) >_V 
    \mu(v) = f_{w,\mu}(i)$
    we get the following two cases:
    if $|u|_{\tb}+2|v|_{\tc}+2t =  |u|_{\tb}+r+t-1+2|v|_{\tc}$
    and $|u|_{\tb}+2|v|_{\ta}+2r > 2|v|_{\ta}+ |u|_{\tb}+r+t-1$
    hold we get $t = r-1$ and $r>t-1$, thus $t+1 = r$.
    If $|u|_{\tb}+2|v|_{\tc}+2t > |u|_{\tb}+r+t-1+2|v|_{\tc}$
    holds we get $t+1 > r$.
    Hence, in general we know $t+1 \geq r$ must hold.
    Now set $u'=u[2..|u|]$ (the case $u'=u[1..|u|-1]$ is symmetric).
    By the assumption that $v$ and $u$ do not overlap
    we have $\mu(u') <_V \mu(v)$.
    We now evaluate this inequality in a similar fashion but also considering
    the three possible letters for $u[1]$.\\
    \textbf{case 1:} $u[1]=\ta$\\
    We have $\mu(u')=
    \binom{ |u'|_{\tb}+2|u'|_{\tc} } { 2|u'|_{\ta}+|u'|_{\tb} }
    =\binom{ |u|_{\tb}+2|u|_{\tc} }  { 2(|u|_{\ta}-1)+|u|_{\tb}}
    =\binom{ |u|_{\tb}+2|v|_{\tc}+2t}{ 2|v|_{\ta}+2r-2+|u|_{\tb}}$.\\
    By $\mu(u')<_A\mu(v)$ we have either
    $|u|_{\tb}+2|v|_{\tc}+2t = |u|_{\tb}+r+t-1+2|v|_{\tc}$
    and $2|v|_{\ta}+2r-2+|u|_{\tb} < 2|v|_{\ta}+|u|_{\tb}+r+t-1$
    which implies $t=r-1$ and $r-1<t$,
    which is a contradiction,
    or $|u|_{\tb}+2|v|_{\tc}+2t < |u|_{\tb}+r+t-1+2|v|_{\tc}$
    which implies $t<r-1$, which is a contradiction to $t+1 \geq r$.
    Hence we get $u[1]\neq\ta$\\
    \textbf{case 2:} $u[1]=\tb$\\
    We have $\mu(u')=
    \binom{ |u'|_{\tb}+2|u'|_{\tc} } { 2|u'|_{\ta}+|u'|_{\tb} }
    =\binom{ |u|_{\tb}-1+2|u|_{\tc} }  { 2|u|_{\ta}+|u|_{\tb}-1}
    =\binom{ |u|_{\tb}-1+2|v|_{\tc}+2t}{ 2|v|_{\ta}+2r+|u|_{\tb}-1}$. 
    By $\mu(u')<_V\mu(v)$ we have either
    $|u|_{\tb}-1+2|v|_{\tc}+2t = |u|_{\tb}+r+t-1+2|v|_{\tc}$
    and $2|v|_{\ta}+2r+|u|_{\tb}-1 < 2|v|_{\ta}+|u|_{\tb}+r+t-1$
    which gives again a contradiction by $t=r$ and $r<t$,
    or $|u|_{\tb}-1+2|v|_{\tc}+2t < |u|_{\tb}+r+t-1+2|v|_{\tc}$
    which implies $t<r$. \\
    \textbf{case 3:} $u[1]=\tc$\\
    We have $\mu(u')=
    \binom{ |u'|_{\tb}+2|u'|_{\tc} }   { 2|u'|_{\ta}+|u'|_{\tb} }
    =\binom{ |u|_{\tb}+2(|u|_{\tc}-1)} { 2|u|_{\ta}+|u|_{\tb}}
    =\binom{ |u|_{\tb}+2|v|_{\tc}+2t-2}{ 2|v|_{\ta}+2r+|u|_{\tb}}$. 
    By $\mu(u')<_A\mu(v)$ we have here either
    $|u|_{\tb}+2|v|_{\tc}+2t-2 = |u|_{\tb}+r+t-1+2|v|_{\tc}$
    and $2|v|_{\ta}+2r+|u|_{\tb} < 2|v|_{\ta}+|u|_{\tb}+r+t-1$
    which leads to the contradiction $t-1=r$ and $r<t-1$,
    or $|u|_{\tb}+2|v|_{\tc}+2t-2 < |u|_{\tb}+r+t-1+2|v|_{\tc}$
    which implies $t<r+1$.\\
    Hence, in all cases we get $t<r+1$ and by $t+1 \geq r$
    we know $t = r-1$ or $t = r$ must hold.
    We can now prove that $\mu$ is gapfree by distinguishing these cases.\\
    \textbf{case 1:} $t=r-1$\\
    By $r+s+t=1$ we get $s=-2t$ and consequently
    \begin{align*}
        f_{w,\mu}(i+1)
        &= f_{w,\mu}(i) \circ_V \binom{s+2t}{2r+s}
         = f_{w,\mu}(i) \circ_V \binom{0}   {2r-2t}\\
        &= f_{w,\mu}(i) \circ_V \binom{0}   {2(r-t)}
         = f_{w,\mu}(i) \circ_V \binom{0}   {2}\\
        &=f_{w,\mu}(i)  \circ_V \mu(\ta).
    \end{align*}
    \textbf{case 2:} $t=r$\\
    By $r+s+t=1$ we get $s=1-2t$ and consequently
    \begin{align*}
        f_{w,\mu}(i+1)
        &= f_{w,\mu}(i) \circ_V \binom{s+2t}{2r+s}
         = f_{w,\mu}(i) \circ_V \binom{1}   {1}\\
        &=f_{w,\mu}(i)+\mu(\tb)
        \text{.}
    \end{align*}
   Thus in both cases exists an $x\in\Sigma$ with 
   $f_{w,\mu}(i+1)=f_{w,\mu}(i)\circ_V\mu(x)$. 
   \qed
\end{proof}

%% file: sec-iwm.tex
%!TEX root = wpnw.tex 
In this section we first show that in the case of an gapfree and injective weight measure
the prefix normal form can be computed deterministically, and non-deterministically if the weight 
measure is not injecitve. Afterwards, we investigate non-injective \wm s, specifically we provide a 
construction that can be used to transform any \wm{} into an injective one.
Thus, w.l.o.g. we always may assume to have an injective \wm.

\begin{lemma}
 Let $\mu$ be a gapfree and injective \wm{}
    over the alphabet $\Sigma$ w.r.t. the monoid $A$ 
    and $w\in\Sigma^*$.
    Then the $\mu$-\pn{} form $w'$ of $w$ can be 
    constructed inductively: $w'[1] = \ta$ if $f_{w,\mu}(1) = \mu(\ta)$ and
    for all $i\in[|w|]$, $i>1$ set $w'[i] = \ta \in \Sigma$ if 
    $f_{w,\mu}(i) = f_{w,\mu}(i-1)\circ_A \mu(\ta)$.
    In contrast, for a \wm{} that is gapfree but \emph{not} injective this
    inductive construction can be used 
    to non-deterministically construct
    all \pn{} words within the 
    \factorWEqnce{} class of a word.
\end{lemma}

\ifpaper
\input{proof_propPNFC}
\else
\fi

\begin{definition}
    \label{def:projWM}
    Let $\mu$ be a \wm{} over $\Sigma$
    w.r.t. the monoid $A$.    
    We define the \emph{$\mu$-projected alphabet} 
    $\Sigma_\mu:= \{ [\ta]_\mu \mid \ta \in\Sigma\} \text{,} $
    where $[\ta]_\mu=\{ \tb\in\Sigma\mid \mu(\tb)=\mu(\ta) \}$
    for $\ta\in\Sigma$ and set $\mu$'s \emph{projected \wm{}} as the 
    \wm{} $\hat{\mu}$ over $\Sigma_\mu$ w.r.t. $A$
    with the base weights
    $\hat{\mu}([\ta]_\mu) = \mu(\ta)\text{.}$
    Finally for a word $w\in\Sigma^*$ we construct its \emph{$\mu$-projection}
    $w_\mu\in\Sigma_\mu^*$ with 
    $w_\mu:=[w[1]]_\mu\dots[w[|w|]]_\mu \text{.}$
\end{definition}

\begin{lemma}
    \label{lem:makeInjective}
    For a weight measure $\mu$ over an alphabet $\Sigma$
    and a word $w\in\Sigma^*$ we have $\hat{\mu}(w_\mu)=\mu(w)$
    and the projected \wm{} $\hat{\mu}$ is injective on $\Sigma_\mu$.
\end{lemma}
\ifpaper % ONLY PAPER
\input{proof_lemMakeInj}
\else    % ONLY ARCHIVE
\fi

%%%%Duplicate from below
%For example again consider the sum weight measure
%$\mu$ over $\Sigma=\{\ta,\tn,\tc,\tb\}$
%with base weights $\mu(\ta) = 1$, $\mu(\tn) = \mu(\tc) = 2$, and $\mu(\tb) = 
%3$.
%Then $\mu$'s projected \wm{} $\hat{\mu}$ is a \wm{} over the alphabet
%$\Sigma_\mu=\{\{\ta\},\{\tn,\tc\},\{\tb\}\}$
%with base weights
%$\hat{\mu}(\{\ta\}) = 1$,
%$\hat{\mu}(\{\tn,\tc\}) = 2$, and
%$\hat{\mu}(\{\tb\}) = 3$
%and we see that $\hat{\mu}$ is injective over $\Sigma_\mu$.
    
\begin{remark}
    With this construction a word $w$ and its $\mu$-projection
    $w_\mu$ behave the same way under any 
    function that is based on the weights of 
    the letters in the words, e.g. 
    $f_{w,\mu} = f_{w_\mu,\hat{\mu}}$,
    $p_{w,\mu} = p_{w_\mu,\hat{\mu}}$,
    $\maxpos_{w,\mu} = \maxpos_{w_\mu,\hat{\mu}}$,
    and $\minpos_{w,\mu} = \minpos_{w_\mu,\hat{\mu}}$ all hold.
    Analogously, all other statements depending on those
    functions hold for the $\mu$-projection of the words as well.
    Thus, $w_\mu$ represents all words within
    the set $\{v\in\Sigma^* \mid v[i]\in w[i] \text{ for all } i\in[|w|]\}$.
\end{remark}

The following theorem essentially shows that 
the \pn{} form of a projected word 
like in Definition~\ref{def:projWM}
represents the set of \pn{} words that are
\factorWEqnt{} to the original word.
In other words, for some $w\in\Sigma^*$ the sets
$\mathcal{P}_\mu(w)$ and $\mathcal{P}_{\hat{\mu}}(w_\mu)$
represent the same \pn{} words over $\Sigma$ 
that are equivalent to $w$.
Thus, also in the non-injective case
we are able to obtain one \pn{} form 
by considering projections.

\begin{theorem}
    \label{the:PiP}
    Let $\mu$ be a gapfree \wm{} over $\Sigma$  
    and let $w\in\Sigma^*$.
    Then with $w'\in \mathcal{P}_{\hat{\mu}}(w_\mu)$ we have
    $\mathcal{P}_\mu(w)
    = \{ v\in\Sigma^* \mid v[i]\in w'[i] \text{ for all }i\in[|w|] \}$.
\end{theorem}
\ifpaper % ONLY PAPER
\input{proof_thePiP}
\else    % ONLY ARCHIVE
\fi

With Theorem~\ref{the:PiP} we can also
accurately calculate the cardinality of $\mathcal{P}_\mu(w)$
for some word $w\in\Sigma^*$ and a non-injective \wm{} $\mu$.

\begin{corollary}
    \label{cor:count}
    Let $\mu$ be a gapfree \wm{} 
    over the alphabet $\Sigma$,
    $w\in\Sigma^*$, and $w'=\mathcal{P}_{\hat{\mu}}(w_\mu)$.
    Then $|\mathcal{P}_\mu(w)|=\prod^{|w|}_{i=1}|w'[i]|$ holds.
\end{corollary}
\ifpaper % ONLY PAPER
\input{proof_corCountP}
\else    % ONLY ARCHIVE
\fi

We conclude this section by revisiting an example w.r.t.
the projected weight measure.
Again consider the \sumwm{} $\mu$ over $\Sigma=\{\ta,\tn,\tc,\tb\}$
with the base weights $\mu(\ta) = 1$, $\mu(\tn) = \mu(\tc) = 2$, and $\mu(\tb) 
= 3$.
Then $\mu$'s projected \wm{} $\hat{\mu}$
is a \wm{} over the alphabet
$\Sigma_\mu=\{\{\ta\},\{\tn,\tc\},\{\tb\}\}$ with the base weights
$\hat{\mu}(\{\ta\}) = 1$,
$\hat{\mu}(\{\tn,\tc\}) = 2$, and
$\hat{\mu}(\{\tb\}) = 3$ and we see that $\hat{\mu}$ is injective on 
$\Sigma_\mu$.
We already know that $\mathtt{nanaba}$ has
multiple \factorWEqnt{} words that are \pn{},
specifically we have
$\mathcal{P}_\mu(\mathtt{nanaba})=
\{\mathtt{banana},\mathtt{bacana},\mathtt{banaca},\mathtt{bacaca}\}$.
Thus, we have the the \pn{} form
$\{\tb\}\{\ta\}\{\tn,\tc\}\{\ta\}\{\tn,\tc\}\{\ta\}$
of 
$(\mathtt{nanaba})_\mu=\{\tn,\tc\}\{\ta\}\{\tn,\tc\}\{\ta\}\{\tb\}\{\ta\}$.
All \factorWEqnt{} and \pn{} words are represented by
this word when reading it as a non-deterministic concatenation of letters,
like shown in Theorem~\ref{the:PiP}, i.e., we have 
$\mathcal{P}_\mu(\mathtt{nanaba})
= \{ v\in\Sigma^* \mid v[i]\in 
\mathcal{P}_{\hat{\mu}}((\mathtt{nanaba})_\mu)[i], i\in[6] \}$.

%% file: proof_propPNFC.tex
\begin{proof}
First, such a $w'$
exists because $\mu$ is 
gapfree and therefore for all $i\in[|w|]$ 
there exists an $\ta \in \Sigma$ with 
$f_{w,\mu}(i) = f_{w,\mu}(i-1) \circ \mu(\ta)$.
Second, $w'$
is unambiguous because 
$\mu$ is injective and therefore there 
exists exactly one $\ta \in \Sigma$ for 
any specific weight.
And third, the word $w'$ is in $[w]_{\sim_\mu}$ 
and is $\mu$-\pn{} because its prefixes are
constructed to have exactly the weight of $w$'s 
corresponding factor with the maximum weight.
So we have $w'=\mathcal{P}_\mu(w)$.
We show this by induction over $i\in[|w|]$.
For $i = 1$, we directly have $p_{w',\mu}(1) = \mu(w'[1]) = f_{w,\mu}(1)$ by construction.
Assuming the claim holds for some $i\in[|w|-2]$,
then we have $p_{w',\mu}(i+1) = \mu(w'[1\dots i+1]) = p_{w',\mu}(i) \circ \mu(w'[i+1])$ which is equal to
$f_{w,\mu}(i) \circ \mu(w'[i+1]) = f_{w,\mu}(i+1)$
by our inductive assumption.

In the non-injective case, all choices of appropriate weights lead to a weighted prefix normal 
form.\qed
\end{proof}

%% file: proof_lemMakeInj.tex
\begin{proof}
    The first statement holds directly by the construction of the
    projected \wm{} with the base weights $\hat{\mu}([\ta]_\mu) = \mu(\ta)$
    for all $\ta\in\Sigma$. For  the second claim choose 
    $[\ta]_\mu, [\tb]_\mu \in \Sigma_\mu$
    with $\hat{\mu}([\ta]_\mu) = \hat{\mu}([\tb]_\mu)$.
    Then $\mu(\ta) = \mu(\tb)$ holds and 
    therefore both $\tb\in[\ta]_\mu$ and $\ta\in[\tb]_\mu$ hold.
    In other words $[\ta]_\mu = [\tb]_\mu$,
    so $\hat{\mu}$ is injective on $\Sigma_\mu$. 
    \qed
\end{proof}

%% file: proof_thePiP.tex
\begin{proof}
	Let $\mu$ be a gapfree \wm{} over $\Sigma$ w.r.t. the monoid $A$ 
	and let $w\in\Sigma^*$.
    Let $w'$ be the \pn{} form of $w$'s $\mu$-projection $w_\mu$, so
    $w':= \mathcal{P}_{\hat{\mu}}(w_\mu)$. This is possible because 
    $\hat{\mu}$ is gapfree and injective
    by Lemma~\ref{lem:makeInjective}.
    Notice here that $|w|=|w_\mu|=|w'|$ holds by construction of $w_\mu$ and $w'$.
%    We now show the claim by showing mutual inclusion.
    
    First, let $v\in \{ u\in\Sigma^* 
    \mid u[i]\in w'[i] \text{ for all }i\in[|w|] \}$.
    To show $v\in\mathcal{P}_{\mu}(w)$ we prove $v$ is in $[w]_{\sim_\mu}$ and $v$ is $\mu$-\pn{}.       
    Firstly, $f_{v,\mu} = f_{w',\hat{\mu}}$
    holds by the choice of $v$,
    $f_{w',\hat{\mu}} = f_{w_\mu,\hat{\mu}}$
    holds by the choice of $w'$ as the \pn{} form of $w_\mu$,
    and $f_{w_\mu,\hat{\mu}} = f_{w,\mu}$
    holds by the construction of $w_\mu$ and $\hat{\mu}$.
    In total $f_{v,\mu} = f_{w,\mu}$ holds
    and we have $v\in[w]_{\sim_\mu}$.   
    Secondly, also $p_{v,\mu} = p_{w',\mu'}$
    holds by the choice of $v$
    and also $p_{w',\hat{\mu}} = f_{w_\mu,\hat{\mu}}$
    holds by the choice of $w'$ as the \pn{} form of $w_\mu$.
    In total $p_{v,\mu} = f_{w_\mu,\hat{\mu}} = f_{w,\mu}$
    holds and so $v$ is $\mu$-\pn{}.
    
    For the second part of the proof, let $v\in\mathcal{P}_{\mu}(w)$.
    To show $v\in \{ u\in\Sigma^* 
    \mid u[i]\in w'[i] \text{ for all }i\in[|w|] \}$
    we prove that $v[i]\in w'[i]$ for all $i\in[|v|]$.
    Let $i\in[|v|]$.        
    By the construction of the \pn{} form
    (Remark~\ref{rem:pnfc}) we know that
    $v[i] = \ta$ for some $\ta\in\Sigma$
    such that $f_{w,\mu}(i)=f_{w,\mu}(i-1)\circ_A\mu(\ta)$ holds.
    We also know $w'[i] = \tx$ for some $\tx\in\Sigma_\mu$
    such that $f_{w_\mu,\hat{\mu}}(i)
    =f_{w_\mu,\hat{\mu}}(i-1)\circ_A\hat{\mu}(\tx)$ holds.
    Since $f_{w_\mu,\hat{\mu}} = f_{w,\mu}$ again
    holds by the construction of $w_\mu$ and $\hat{\mu}$,
    we have $\hat{\mu}(\tx) = \mu(\ta)$.
    So in total, because $\tx$ is the subset of $\Sigma$
    with all the letters of weight $\hat{\mu}(\tx)$,
    by the construction of $\Sigma_\mu$, we have $\ta\in\tx$,
    i.e. $v[i]\in w'[i]$.
    \qed
\end{proof}

%% file: proof_corCountP.tex
\begin{proof}
    Follows directly by Proposition~\ref{the:PiP}.
    \qed
\end{proof}

%% file: section-xpnw.tex
In this section we briefly investigate a na\"ive
approach to generalise binary \pnity{} and
prove that it is already covered by the \wm{} approach.
The main idea is if $\Sigma$ is a finite alphabet to take
a subset $X\subseteq\Sigma$ and instead of counting the 
amount of $\tone$ or $\tzero$ respectively we count how many
letters in a prefix or factor are contained in $X$. Therefore, We generalise the notation 
$|w|_{\ta}$ for a letter $\ta\in\Sigma$ to 
$|w|_X := \abs{\{i\in[|w|] \mid w[i]\in X\}}$
for $X\subseteq\Sigma$, i.e. $|w|_X$ is 
the number of letters of $w$ that are elements of $X$.

%Def f and p
\begin{definition}
    \label{def:Xfs}
    Let $w\in\Sigma^*$ and $X\subseteq\Sigma$.
    We define the \emph{\prefXf} $p_{w,X}$ and the \emph{\factXf} $f_{w,X}$ respectively by
    $p_{w,X}: [|w|]_0 \to \N, i\mapsto |\Pref_i(w)|_X$ and $f_{w,X}:[|w|]_0 \to \N,
    i\mapsto \max(|\Fact_i(w)|_X)$.
    We say that $w$ is \emph{$X$-\pn{}} (or subset \pn{} w.r.t $X$)
    if $p_{w,X} = f_{w,X}$ holds.
\end{definition}

We now show that subset \pnity{} is indeed a generalisation of
binary \pnity{}, and also that subset \pnity{} can already be 
expressed by means of weighted \pnity{}.
However this is not possible the other way around.
So in total we see that weighted \pnity{} is 
more expressive and therefore a more useful generalisation.

\begin{theorem}
    \label{the:subToBin}
    Binary \pnity{} is expressible by
%    the concept of
    subset \pnity.
    (I.e. there exists $X\subseteq\{\tzero,\tone\}$ such that 
    $X$-\pnity{} is equivalent to binary \pnity{}.)
\end{theorem}
\ifpaper % ONLY PAPER
\else    % ONLY ARCHIVE
\input{proof_theSubToBin}
\fi

In other words, in the context of the binary alphabet 
$\{\tone\}$-\pnity{} and \pnity{} are the same.
\begin{theorem}
    \label{the:wToX}
    Subset \pnity{} is expressible by
%    the concept of
    weighted \pnity{}.
    (I.e. for every $X\subseteq \Sigma$ there exists a \wm{} $\mu$ such that 
    $\mu$-\pnity{} is equivalent to $X$-\pnity{}.)
\end{theorem}
\ifpaper % ONLY PAPER
\else    % ONLY ARCHIVE
\input{proof_theWToX}
\fi

By Theorem~\ref{the:wToX} we immediately see that subset \pnity{}
behaves exactly like weighted \pnity{} when using a
binary \wm{}, which we know by Lemma~\ref{lem:binwm} is gapfree.

%% file: proof_theSubToBin.tex
\begin{proof}
    W.l.o.g. consider just $\tone$-\pnity{} for the binary case.
    We choose $X\subseteq\Sigma$ with $X=\{\tone\}$.
    Then $|w|_\tone= |w|_X$ holds
    for any binary word $w\in\{\tzero,\tone\}^*$.
    It follows that
    $f_w(i)= \max(|\Fact_i(w)|_\tone)
    = \max(|(\Fact_i(w)|_X) = f_{w,X}(i)$
    and $p_w(i)= p_{w,X}(i)$ hold
    for all $i \in [|w|]$.
    Therefore, $w$ is $X$-\pn{} 
    if and only if it is $\tone$-\pn.
    So, with such an $X$ every statement on
    binary \pnity{} can be transformed
    into an analogue using subset \pnity{}.
    \qed
\end{proof}

%% file: proof_theWToX.tex
\begin{proof}
    Let $\Sigma$ be an alphabet and let $X\subseteq\Sigma$.
    We construct a \sumwm{} $\mu$ over $\Sigma$.
    Let $\mu(\tx)=2$ and $\mu(\ty)=1$ for 
    every $\tx\in X$ and $\ty\in\Sigma\backslash X$.
    Then $|w|_\tone + |w| = \mu(w)$ holds
    for any word $w\in\Sigma$.
    It follows that
    $f_{w,X}(i) + |w| = \max(|Fact_i(w)|_X) + |w|
    = \max(\mu(Fact_i(w))) = f_{w,\mu}(i)$
    and $p_{w,X}(i) + |w| = p_{w,\mu}(i)$ hold
    for all $i \in [|w|]$.
    Therefore, $w$ is $\mu$-\pn{} 
    if and only if it is $X$-\pn.
    So, with such a \wm{} every statement on
    subset \pnity{} can be transformed
    into an analogue using weighted \pnity{}.
    \qed
\end{proof}

%% file: appendix-proofs.tex
\noindent
\textbf{Proof of Proposition~\ref{the:weight2bin}.}

\input{proof_propBinByWeight}
\noindent
\textbf{Proof of Theorem~\ref{the:P}.}

\input{proof_theWPNF}
\noindent
\textbf{Proof of Lemma~\ref{lem:binwm}.}

\input{proof_lemBinWms}
\noindent
\textbf{Proof of Lemma~\ref{lem:primeWM}.}

\input{proof_lemPrimeWM}
\noindent
\textbf{Proof of Proposition~\ref{prop:step1}.}

\input{proof_propSTEP1}
\noindent
\textbf{Proof of Proposition~\ref{prop:EQWMEQPNF}.}

\input{proof_propEQWMPNF}
\clearpage

\noindent
\textbf{Proof of Lemma~\ref{prop:EQWMS}.}

\input{proof_propEQWMS}
\noindent
\textbf{Proof of Lemma~\ref{lem:STDWM}.}

\input{proof_lemSTDWM}
\noindent
\textbf{Proof of Lemma~\ref{lem:cacbtok2}.}

\input{proof_lemcacbtok2}
\noindent
\textbf{Proof of Theorem~\ref{the:threeEQ}.}

\input{proof_thethreeEQ}
%
%\noindent
%\textbf{Proof of Lemma~\ref{lem:makeInjective}.}
%\input{proof_lemMakeInj}
%
%\noindent
%\textbf{Proof of Theorem~\ref{the:PiP}.}
%\input{proof_thePiP}
%
%\noindent
%\textbf{Proof of Corollary~\ref{cor:count}.}
%\input{proof_corCountP}